\documentclass[11pt]{amsart}
\usepackage{harvard,enumerate,amsmath,amssymb,enumitem}
\newtheorem{cor}{Corollary}[section]
\newtheorem{lem}{Lemma}[section]

\newtheorem{ass}{Assumption}[section]

\theoremstyle{definition}

\newtheorem{rem}{Remark}[section]

\numberwithin{equation}{section}

\newcommand{\E}{\mathbb{E}}
\newcommand{\R}{\mathbb{R}}

\newcommand{\argmax}{\operatornamewithlimits{argmax}}
\newcommand{\argmin}{\operatornamewithlimits{argmin}}

\usepackage{color}

\DeclareMathOperator*{\supp}{supp} 
\DeclareMathOperator*{\tr}{tr} 
\newcommand{\N}{\mathbb{N}}

{

\appendix
\renewcommand{\thesection}{\Alph{section}}

}

 \allowdisplaybreaks
 
\begin{document}
	\title[]{Estimating Stochastic Block Models in the Presence of Covariates
	}

	\author[Kitamura]{Yuichi Kitamura}
	\address{Cowles Foundation for Research in Economics, Yale University, New
		Haven, CT 06520.}
	\email{yuichi.kitamura@yale.edu}     
	\author[Laage]{Louise Laage}
	\address{Department of Economics, Georgetown University, Washington, D.C., 20007.}
	\email{louise.laage@georgetown.edu}

	\date{This Version: Feb 25, 2024.} 
	
	
	\begin{abstract} 
			
				In the standard stochastic block model for networks, the probability of a connection 
			between two nodes, often referred to as the edge probability, depends on the unobserved 
			communities each of these nodes belongs to.  We consider a flexible framework in which each edge
			probability, together with the probability of community assignment, are also impacted by 
			observed covariates. We propose a computationally tractable two-step procedure to estimate 
			the conditional edge probabilities as well as the community assignment probabilities. The first step relies on a 
			spectral clustering algorithm applied to a localized adjacency matrix of the network. In the
			second step, $k$-nearest neighbor regression estimates are computed on the extracted 
			communities. We study the statistical properties of these estimators by providing 
			non-asymptotic bounds.
			
	\end{abstract}
	
	\maketitle

\section{Introduction}\label{sec:intro}

The Stochastic Block Model (SBM) is a powerful yet convenient framework  for network data analysis, by postulating that each node in a network belongs to a community, in the presence of  a finite number of communities.  The community assignments are unobserved to the researcher.    A SBM is highly effective  in applications where community membership can be regarded as  a discretized version of unobserved heterogeneity.    The standard SBM has been applied  widely in diverse areas, where  an algorithm based on spectral clustering is often used.  Statistical  properties of  these methodologies have been investigated intensively in the recent literature; see, \citeasnoun{lei2015consistency}, \citeasnoun{joseph2016impact} and \citeasnoun{rohe2016co}, just to name a few.   

As an alternative approach, one can employ a model  with more specific structures for the edge probability functions, while  incorporating (possibly continuous-valued) unobserved heterogeneity through node-specific fixed effects.  An advantage of such an approach is its ease of incorporating  observed covariates into the model.  If one is willing to accept a specific form of the  edge probability function, e.g. scalar valued fixed effects representing unobserved heterogeneity, and the estimation algorithm is computationally feasible, then such a method is practical and intuitive to use, and by having covariates in the model, it can shed valuable insights on, for example, the magnitude of homophily (or heterophily) effects in terms of observed characteristics in a given network.  

The main goal of our paper is to develop a procedure  that incorporates  both unobserved heterogeneity (through unobserved community assignments) and  observed heterogeneity (through covariates) building on the SBM framework.  More specifically, our version of SBM lets each edge
probability, together with the probability of community assignment, be impacted by 
observed covariates in a flexible manner.  That is, we postulate that the edge probabilities and  the community assignment probabilities are nonparametric function of covariates.  Unobserved community assignment enters the model in a fully unrestricted way, since it simply indexes (or used as a label for) the edge probability  and the community assignment probability.   We note that letting the community assignment probability depend on covariates nonparametrically is an important feature of our model.  It means that the observed covariates (or observed heterogeneity)  and the unobserved heterogeneity can be correlated in an unspecified way, a feature often considered to be highly desirable in econometrics.  Of course, this generality comes at the cost of additional technical complications, a part of  the many theoretical challenges presented by our model, as is discussed shortly.     

We propose a computationally tractable two-step procedure to estimate 
the conditional edge probabilities as well as the community assignment probabilities. The first step relies on a 
spectral clustering algorithm applied to a localized adjacency matrix of the network.  We build on the $k$-nearest neighbor ($k$-nn) algorithm in our localization procedure, followed by the Singular Value Decomposition (SDV) to compute left/right singular vectors of the localized adjacency matrix, or, more precisely, a localized and normalized version of the Laplacian.  Note that we need to employ SVD, as opposed to the  eigen-value decomposition often used for the standard spectral clustering, even though we consider undirected network.  This is because our localized and normalized network Laplacian  depends on the covariate values at each of the two nodes, resulting its asymmetry.    We then apply $K$-means clustering to extract communities, as in  the standard spectral clustering algorithm, though it is implemented  at each covariate value.    As in the literature of the standard SBM without covariates, we provide non-asymptotic bounds for  misclassificaiton error.        Since we allow for correlation of unknown form  between the community assignment and the covariates, it naturally induces independent but non-identically distributed (i.n.i.d) covariates when conditioned on community assignments, even though we assume random sampling for community assignments and observed covariates.  We thus obtain some non-asymptotic results for  the $k$-nn estimator under an i.n.i.d sampling; those can be of independent theoretical interest.

 In the
second step, once again we employ the $k$-nn regression algorithm,  but this time in order to estimate the edge probability matrix and the community assignment probability vector.   We apply the $k$-nn estimator  to the extracted communities obtained from the first step.    We study the statistical properties of these estimators while taking account of the effects of classification error in the first step, by obtaining 
non-asymptotic bounds for them.

\subsection{Notation.} 

\

For a matrix $M$, $\|M\|$ denotes its spectral norm, $\|M\|_F$ its Frobenius norm, $\|M\|_{\max}$ its max norm. $\mathcal{B}(x,r) \subset \R^{d}$ is the closed ball of center $x$ and radius $r$, $\lambda$ is the Lebesgue measure on $\R^d$,  $V_d = \int_{\mathcal{B}(0,1)}\text{d}\lambda$ and $I_k$ is the identity matrix of size $k$.

\section{Model and Estimator}

\subsection{Stochastic Block Model with Covariates}

\

We consider a network of $N$ nodes such that each node $i$ has a $d$-vector of covariates $x(i)$ and belongs to a community $g(i) \in [G]$ where $G \in \N$. For each node $i$, the researcher observes  $x(i)$ but does not observe the community $g(i)$.  Let $\mathbb{M}_{N,G} \subset \R^{N \times G}$ the set of membership matrices, i.e., of matrices such that each row has exactly one nonzero coefficient set to $1$. Let $\theta_i \in \R^G$ be a vector such that $\theta_{i g(i)} = 1$ and all other coefficients are $0$. Then $\Theta := (\theta_1, ..., \theta_N)^{\top} \in \mathbb{M}_{N,G}$.
The researcher also observes the adjacency matrix $A=\left(A_{ij}\right)_{1\leq i,j\leq N}$. Define ${\bf x} = \{x(i), i \in [N]\}$ and ${\bf g} = \{g(i), i \in [N]\}$.
In our stochastic block model with covariates, the distribution of the adjacency matrix conditional on ${\bf x}$ and ${\bf g}$
is given by a matrix-valued function $B: (x,x') \mapsto B(x,x') \in \R^{G\times G}$: conditional on ${\bf x}$ and ${\bf g}$, the entries $A_{ij}$ are i.i.d Bernouilli random variables with 
$$
\Pr(A_{ij} =1 | {\bf x}, {\bf g}) = \Pr(A_{ij} =1 | x(i),x(j),g(i),g(j)) = B_{g(i)g(j)}(x(i),x(j)).
$$
We note that the function $B$ may vary with $N$ and thus allows for sparsity. The distribution of $g(i)$ conditional on $x(i)$ is given by the functions $x \mapsto \pi_g(x) = \Pr(g(i)=g|x(i)=x)$, for $g\in[G]$. We assume that $(x(i),g(i))_{i \in [N]}$ are i.i.d random variables, and let $\mathcal S_X = \supp(x(i))$.    However we state some of our results as nonasymptotic bounds conditional on ${\bf g}$, in which case it is understood that we only assume that $(x(i))_{i \in [N]}$ are independent draws and the marginal distribution of each depend on $g(i)$ only.

\subsection{Construction of the estimator}\label{sec:defEstim}

\

Our primary parameters of interest are $B(x,x')$ and $\pi(x) = (\pi_1(x),...,\pi_G(x))$ for some $(x,x') \in \mathcal{S}_X \times \mathcal{S}_X$. Our estimators of these quantities rely on a spectral clustering algorithm applied to a truncated adjacency matrix, following the intuition of $k$-nearest neighbor regression.
Before describing the algorithm and our estimators, we introduce the following notations. 

\medskip

Let $k \in \N$. 
We define the $k$-nearest neighbor ($k$-NN) radius of $x$ as 
$r_k(x) := \inf \{ r>0 \, : \, |\mathcal{B}(x,r) \cap {\bold x} | = k \}$. The $k$-neighborhood of $x$  is 
$
\eta_N(x) := \{ i \in [N] \, : \, ||x(i) - x|| \leq r_k(x)\}.
$
Let $\eta_N(x,x') := \{(x(i),x(j)): x(i) \in \eta_N(x) \text{ and } x(j) \in \eta_N(x') \}$. For ease of notation, we sometimes write $\eta(x)$. 
For these units, we define the adjacency matrix ${A^{\eta}} \in \R^{k \times k}$: it is a submatrix of the network adjacency matrix $A$ where rows index units in $\eta(x)$ and columns index units in $\eta(x')$.
Similarly,
$\Theta^{\eta}(x) \in \mathbb{M}_{k,G}$ is the membership matrix for the individuals in $\eta(x)$: it is a submatrix of $\Theta$ where rows index units in $\eta(x)$ and columns index communities. Note that we drop the dependence of ${A^{\eta}}$ in $x$ and $x'$ for ease of notation. Let  $O^\eta \in  \R^{k \times k}$ and $Q^\eta \in  \R^{k \times k}$ be the diagonal matrices such that $O_{ii}^\eta := \sum_{j \in \eta_N(x')} A_{ij}^\eta$ and $Q_{jj}^\eta := \sum_{i \in \eta_N(x)} A_{ij}^\eta$. For $\tau > 0$, we also define $O_\tau^\eta := O^\eta + \tau I_{k}$ and $Q_\tau^\eta := Q^\eta + \tau I_{k'}$.  Finally, by analogy with the regularized graph Laplacian in the symmetric case, let 
$$L_\tau^\eta : = {(O_\tau^\eta)}^{-\frac 1 2} A^\eta  {(Q_\tau^\eta)}^{- \frac 1 2}. $$
Our estimation procedure first applies a spectral clustering algorithm to $L_\tau^\eta$ to estimate the communities of units in $\eta(x)$ and $\eta(x')$. This spectral clustering algorithm proceeds as follows.

\medskip

\textit{\textbf{Spectral Clustering Algorithm:}}

\begin{itemize}
	\item[] \textit{\textbf{Input: }} $L_\tau^\eta$, $G$, approximation parameter for $k$-means.
	
	\item[] \textit{\textbf{Output: }} $\widehat{\Theta}^\eta(x) \in \mathbb{M}_{k,G}$ and $\widehat{\Theta}^\eta(x') \in\mathbb{M}_{k,G}$ estimators of the membership matrices.
	
	\item[] \textit{\textbf{Steps: }} 
	
	\begin{enumerate}
		\item Obtain the singular value decomposition of $L_\tau^\eta$ with $L_\tau^\eta = \sum_{g=1}^{G} \lambda_g(x,x') U_g V_g'$. Let $U \in \R^{k\times G}$ and  $V \in \R^{k\times G}$ be the matrices of the top $G$ left and right singular vectors, respectively.
				
		\item Apply $K$-means clustering on the rows of $U$, output $\widehat{\Theta}^\eta(x)$.
		
		\item Apply $K$-means clustering on the rows of $V$, output $\widehat{\Theta}^\eta(x')$.		
	\end{enumerate}
\end{itemize}

\medskip

Once these communities are estimated, we estimate $B_{gh}(x,x')$ and $\pi_g(x)$ by running $k$-NN regressions on the appropriate communities. For $i \in \eta(x)$, let $\hat{g}(i)$ be the estimated community, i.e., such that $\left(\widehat{\Theta}^\eta(x)\right)_{i, \hat{g}(i)} = 1$ and $\left(\widehat{\Theta}^\eta(x)\right)_{i,g} = 0$ for any $g \neq \hat{g}(i)$. Estimated communities of units in $\eta(x')$ are similarly defined. Let  $\mathcal{G}_{h N}(x) = \{g(i) = h, i \in \eta_N(x)\}$ and $n_h(x) = |\mathcal{G}_{h N}(x)|$.  Our estimator for $n_h(x)$ is $\hat{n}_h(x) := \# \{i \in \eta(x) \, : \, \hat{g}(i) = h  \}$.
We can finally introduce our estimators of the connection probabilities and community probabilities. The estimator of $\pi_h(x)$ is
\begin{equation}
	\widehat{\pi}_h(x)  = \frac{\hat{n}_h(x)}{k}.  \label{eq:hat_prop}
\end{equation}
The estimator of $B_{gh}(x,x')$ is
\begin{equation}
	\widehat{B}_{gh}(x,x')  = \frac{1}{\hat{n}_g(x) \hat{n}_h(x')} \ \sum_{\substack{i \in \eta(x) \, : \, \hat{g}(i)=g\\ j \in \eta(x') \, : \, \hat{g}(j)=h}}  A_{ij}. \label{eq:hat_graphon} 
\end{equation}

\subsection{Assumptions}\label{sec:assn}

\

We will maintain the following assumptions where $S$ is a subset of $\mathcal{S}_X$, the support of $x$.
\begin{ass}\label{assn:boundsupport}
	There exist  constants  $c>0$ and $T>0$  such that
	\begin{equation}
		\lambda(S \cap \mathcal{B}(x,\tau)) \geq c \lambda(\mathcal{B}(x,\tau)), \ \forall \tau \in (0,T], \, \forall x \in S. \label{eq:boundsupport}
	\end{equation}
\end{ass}
We consider the case where the covariates are continuous. Define the density of $x(i)$ given $g(i)=g$ as $f(.|g)$,  for $g \in [G]$. Define also
$\underline{f}(x) := \min_{g \in [G]} \ f(x|g) $. 
\begin{ass}\label{assn:boundinff} 
	There exist constants  $U_X$ and $b_X >0$   such that
	\begin{align}
		U_X \geq \underline{f}(x) \geq& b_X, \ \forall \, x \in S. \label{eq:boundinff} 
	\end{align}
\end{ass} 
Similarly, define 
$
\overline{f}(x) := 
\max_{g \in [G]} \ f(x|g).
$
\begin{ass}\label{assn:boundsupf}
	There exist constants $\overline{U}_X$ and $\overline{b}_X > 0$ such that 
	\begin{equation}\label{eq:boundsupf}
		\overline{U}_X \geq \overline{f}(x) \geq \overline{b}_X, \ \forall \, x \in S.
	\end{equation}
\end{ass}
Some smoothness assumptions will be imposed to study our estimators.
\begin{ass}\label{assn:LipshitzB}
	$(x,x') \mapsto B_{gh}(x,x')$ is Lipschitz continuous for all $(g,h) \in [G]^2$. We denote  $l_B$ the smallest Lipschitz constant.
\end{ass} 
\begin{ass}\label{assn:LipshitzPi}
	$x \mapsto \pi_{g}(x)$ is Lipschitz continuous for all $g\in[G]$. We denote $l_\pi$ the smallest Lipschitz constant. 
\end{ass}
Note that $l_B$ may depend on $N$ and thus captures sparsity.

\section{Clustering}

We fix a pair $(x,x') \in S^2$. The main result of this section is a high probability  finite sample bound on a misclustering measure applied to  $\widehat{\Theta}^\eta(x)$ and  $\widehat{\Theta}^\eta(x')$. This result is obtained in 3 steps. We first obtain a finite sample bound on the difference between the graph Laplacian and its population counterpart. We then apply a version of Davis-Kahan theorem to bound the difference between $\hat{U}$, $\hat{V}$, and their population counterparts. Finally, we use theoretical properties of the $K$-means algorithm to bound the misclustering error.

\subsection{Convergence of the graph Laplacian}

\

To introduce the population counterpart of $L_\tau^\eta$, we define the following objects. Let $P({\bf x},{\bf g}) := {\mathbb E}[A|{\bf x},{\bf g}]$. Thus, $P_{ij}({\bf x},{\bf g}) = {\mathbb E}[A_{ij}|x(i),x(j),g(i),g(j)]$. Let also $P_{ij}(x,x',{\bf g}) = {\mathbb E}[A_{ij}|x(i) = x,x(j) = x',g(i),g(j)] $ and  $P(x,x',{\bf g}) = \left(P_{ij}(x,x',{\bf g})\right)_{1\leq i,j \leq N}$. As in Section \ref{sec:defEstim}, we also define the localized matrices $P^\eta({\bf x},{\bf g})$ and $P^\eta(x,x',{\bf g})$ as submatrices of the matrices $P({\bf x},{\bf g})$ and $P(x,x',{\bf g})$ where rows index units in $\eta(x)$ and columns index units in $\eta(x')$.
Let  $\mathcal{O}^\eta({\bf x},{\bf g}) \in  \R^{k \times k}$, $\mathcal{O}^\eta(x,x',{\bf g}) \in  \R^{k \times k}$, $\mathcal{Q}^\eta({\bf x},{\bf g}) \in  \R^{k \times k}$ and $\mathcal{Q}^\eta(x,x',{\bf g}) \in  \R^{k \times k}$ be the diagonal matrices such that
${\mathcal O}_{ii}^\eta({\bf x},{\bf g}) := \sum_{j \in \eta_N(x')   } {\mathbb E}[A_{ij}^\eta|{\bf x},{\bf g}]$, ${\mathcal O}_{ii}^\eta(x,x',{\bf g}) := \sum_{j \in \eta_N(x')   }   P_{ij}(x,x',{\bf g})$, ${\mathcal Q}_{jj}^\eta({\bf x},{\bf g}) := \sum_{i \in \eta_N(x)} {\mathbb E}[A_{ij}^\eta|{\bf x},{\bf g}]$ and finally ${\mathcal Q}_{jj}^\eta(x,x',{\bf g}) := \sum_{i \in \eta_N(x)} P_{ij}(x,x',{\bf g})$. Define also 
$\mathcal{O}_\tau^\eta({\bf x},{\bf g}) := \mathcal{O}^\eta({\bf x},{\bf g})+ \tau I_{k}$, $\mathcal{O}_\tau^\eta(x,x',{\bf g}) := \mathcal{O}^\eta(x,x',{\bf g})+ \tau I_{k}$, $\mathcal{Q}_\tau^\eta({\bf x},{\bf g}) := \mathcal{Q}^\eta({\bf x},{\bf g})+ \tau I_{k}$ and $\mathcal{Q}_\tau^\eta(x,x',{\bf g}) := \mathcal{Q}^\eta(x,x',{\bf g})+ \tau I_{k}$. Finally, let
\begin{align*}
	{\mathcal L}_\tau^\eta({\bf x},{\bf g}) & := ({\mathcal O}_\tau^\eta({\bf x},{\bf g}))^{-\frac 1 2} P^\eta({\bf x},{\bf g})   ({\mathcal Q}_\tau^\eta({\bf x},{\bf g}))^{-\frac 1 2},\\
	{\mathcal L}_\tau^\eta(x,x',{\bf g}) & := ({\mathcal O}_\tau^\eta(x,x',{\bf g}))^{-\frac 1 2} P^\eta(x,x',{\bf g})   ({\mathcal Q}_\tau^\eta(x,x',{\bf g}))^{-\frac 1 2}.
\end{align*}	
To obtain a finite sample bound, we will rely on concentration inequalities for symmetric matrices. However, $P^\eta(x,x',{\bf g}) $ and $A^\eta$ are not symmetric in general. Following \citeasnoun{rohe2016co}, we first focus on their Hermitian dilations,  where  the Hermitian dilation of a matrix $M$, denoted $\widetilde{M}$, is given by
$$
\widetilde{M} = 
\begin{pmatrix}
	0 & M\\
	M^\top & 0
\end{pmatrix}.
$$
Define the minimum degree
$$d_{\min}(x,x', {\bold x}) := \min \left(\min_{i \in \eta_N(x)} {\mathcal O}_{ii}(x,x',{\bf g}), \min_{j \in \eta_N(x')} {\mathcal Q}_{jj}(x,x',{\bf g}),\min_{i \in \eta_N(x)} {\mathcal O}_{ii}({\bf x},{\bf g}), \min_{j \in \eta_N(x')} {\mathcal Q}_{jj}({\bf x},{\bf g})\right).$$
In most of the computations, we drop the dependence in $(x,x', {\bold x})$. 
Define 
\begin{align}
	R_{k} &:= \left(\frac{2k}{N b_X c V_d}\right)^{1/d}\label{eq:maxradius}.
\end{align}
\begin{lem}\label{lem:boundLaplacians}
	Let Assumptions \ref{assn:boundsupport}, \ref{assn:boundinff}, \ref{assn:boundsupf} and \ref{assn:LipshitzB} hold. For all $N$, $\delta \in (0,1)$, $\tau >0$ and $1 \leq k \leq N$ such that 
	\begin{enumerate}
		\item \label{cond:radius} $\sup_{x \in S} r_k(x) \leq R_k$,
		\item \label{cond:mindeg} $3 \ln (8k/\delta)\leq d_{\min} + \tau$,
	\end{enumerate}
	with probability at least $1-\delta$ conditional on $({\bold g},{\bold x})$, it holds that
	\begin{align}\label{eq:condLapl}
		 \|\widetilde{L_\tau^\eta} - \widetilde{{\mathcal L}_\tau^\eta}(x,x',{\bf g}) \|  \leq   4\sqrt{\frac{3 \ln (8k/\delta)}{d_{\min} + \tau}} + \frac{2 k l_B R_k}{ d_{\min} + \tau} \left( \frac{2 k l_B R_k}{ d_{\min} + \tau} + 3\right).
	\end{align}
\end{lem}
Note that this finite sample bound holds conditionally on $({\bold g},{\bold x})$. It cannot be integrated directly because Conditions (\ref{cond:radius}) and (\ref{cond:mindeg}) are restrictions on $\bold x$. We show in Lemma \ref{lem:radius} that under certain assumptions, $\sup_{x \in S} r_k(x) \leq R_k$ holds with high probability and we derive in Lemma \ref{lem:bdlocaldegree} a high probability bound for $d_{\min}$. Lemma \ref{lem:integr} obtains a high probability bound on $\|\widetilde{L_\tau^\eta} - \widetilde{{\mathcal L}_\tau^\eta}(x,x',{\bf g}) \|$ which will hold conditional on $\bold g$ only.

\begin{proof}
	
	\
	
	Let
	$$D^\eta_\tau := \begin{pmatrix}
		O_\tau^\eta &  0  
		\\
		0   &  Q_\tau^\eta
	\end{pmatrix}, \ {\mathcal D}^\eta_\tau({\bf x},{\bf g}):= \begin{pmatrix}
		{\mathcal O}_\tau^\eta({\bf x},{\bf g}) &  0  
		\\
		0   &  {\mathcal Q}_\tau^\eta({\bf x},{\bf g})
	\end{pmatrix} \text{ and } {\mathcal D}^\eta_\tau(x,x',{\bf g}):= \begin{pmatrix}
		{\mathcal O}_\tau^\eta(x,x',{\bf g}) &  0  
		\\
		0   &  {\mathcal Q}_\tau^\eta(x,x',{\bf g})
	\end{pmatrix}.$$
	Note that 
	\begin{align*}
		\widetilde{L^\eta_\tau} &= (D_\tau^\eta)^{-\frac 1 2} \widetilde{A^\eta} (D_\tau^\eta)^{-\frac 1 2},\\
		\widetilde{{\mathcal L}_\tau^\eta}({\bf x},{\bf g}) & = {\mathcal D}^\eta_\tau({\bf x},{\bf g})^{-\frac 1 2} \widetilde{P^\eta}({\bf x},{\bf g}) {\mathcal D}^\eta_\tau({\bf x},{\bf g})^{-\frac 1 2}\\
		\widetilde{{\mathcal L}_\tau^\eta}(x,x',{\bf g}) & = {\mathcal D}^\eta_\tau(x,x',{\bf g})^{-\frac 1 2} \widetilde{P^\eta}(x,x',{\bf g}) {\mathcal D}^\eta_\tau(x,x',{\bf g})^{-\frac 1 2}.
	\end{align*}
	We decompose
	\begin{eqnarray*}
	\widetilde{L_\tau^\eta} - \widetilde{{\mathcal L}_\tau^\eta}(x,x',{\bf g}) &=& (D_\tau^\eta)^{-\frac 1 2} \widetilde{A^\eta} (D_\tau^\eta)^{-\frac 1 2} -       ({\mathcal D}^\eta_\tau(x,x',{\bf g}) )^{-\frac 1 2} \widetilde{P^\eta}(x,x',{\bf g})   ({\mathcal D}^\eta_\tau(x,x',{\bf g}) )^{-\frac 1 2} 
	\\
	&=& \underbrace{
		[(D_\tau^\eta)^{-\frac 1 2} -   ({\mathcal D}^\eta_\tau({\bf x},{\bf g}) )^{-\frac 1 2}]\widetilde{A^\eta} (D_\tau^\eta)^{-\frac 1 2}
	}_{=B_1} 
	+ 
	\underbrace{   ({\mathcal D}^\eta_\tau({\bf x},{\bf g}) )^{-\frac 1 2}\widetilde{A^\eta} 
	[(D_\tau^\eta)^{-\frac 1 2} -   ({\mathcal D}^\eta_\tau({\bf x},{\bf g}) )^{-\frac 1 2}]
	}_{= B_2}
	\\
	&&+   
	\underbrace{
	({\mathcal D}^\eta_\tau({\bf x},{\bf g}) )^{-\frac 1 2}\widetilde{A^\eta} 
	({\mathcal D}^\eta_\tau({\bf x},{\bf g}) )^{-\frac 1 2} -    ({\mathcal D}^\eta_\tau({\bf x},{\bf g}) )^{-\frac 1 2} \widetilde{P^\eta}({\bf x},{\bf g})   ({\mathcal D}^\eta_\tau({\bf x},{\bf g}) )^{-\frac 1 2}  
	}_{= B_3}
	\\
	&&+  
	\underbrace{
	   ({\mathcal D}^\eta_\tau({\bf x},{\bf g}) )^{-\frac 1 2} \widetilde{P^\eta}({\bf x},{\bf g})   ({\mathcal D}^\eta_\tau({\bf x},{\bf g}) )^{-\frac 1 2}    
	  - ({\mathcal D}^\eta_\tau(x,x',{\bf g}) )^{-\frac 1 2} \widetilde{P^\eta}(x,x',{\bf g})   ({\mathcal D}^\eta_\tau(x,x',{\bf g}) )^{-\frac 1 2}
	  }_{=B_4}. 
	\end{eqnarray*}
	Note that by Lemma \ref{lem:normL}, $\| \widetilde{L^\eta_\tau}  \|  \leq 1$. Thus
	\begin{eqnarray*}
		\|B_1\| &=& 	\|[I -   ({\mathcal D}^\eta_\tau({\bf x},{\bf g}) )^{-\frac 1 2}(D_\tau^\eta)^{\frac 1 2}](D_\tau^\eta)^{-\frac 1 2}\widetilde{A^\eta} (D_\tau^\eta)^{-\frac 1 2}\|
		\\
		&=&
			\|I -   ({\mathcal D}^\eta_\tau({\bf x},{\bf g}) )^{-\frac 1 2}(D_\tau^\eta)^{\frac 1 2}\widetilde{L^\eta_\tau} \|
			\\
			&\leq & 	\|I -   ({\mathcal D}^\eta_\tau({\bf x},{\bf g}) )^{-\frac 1 2}(D_\tau^\eta)^{\frac 1 2}\|.        
	\end{eqnarray*} 
	Likewise 
	\begin{eqnarray*}
		\|B_2\| &=& 
		\| ({\mathcal D}^\eta_\tau({\bf x},{\bf g}) )^{-\frac 1 2} (D_\tau^\eta)^{\frac 1 2}  (D_\tau^\eta)^{-\frac 1 2}     \widetilde{A^\eta} 
		(D_\tau^\eta)^{-\frac 1 2}[I  -   (D_\tau^\eta)^{\frac 1 2}({\mathcal D}^\eta_\tau({\bf x},{\bf g}) )^{-\frac 1 2}]\|
		\\
		&\leq& \|({\mathcal D}^\eta_\tau({\bf x},{\bf g}) )^{-\frac 1 2} (D_\tau^\eta)^{\frac 1 2} \|        
		\|I  -   (D_\tau^\eta)^{\frac 1 2}({\mathcal D}^\eta_\tau({\bf x},{\bf g}) )^{-\frac 1 2}   \|.  
	\end{eqnarray*}  
	We follow \citeasnoun{rohe2016co} and use the two-sided Chernoff concentration inequality of \citeasnoun{chung2006complex} (see Theorem 2.4) to obtain for $i \in \eta_N(x) \cup \eta_N(x')$,
	$$
	\Pr\left( | ({\mathcal D}^\eta({\bf x},{\bf g}) )_{ii} - (D^\eta)_{ii} | \geq a_0 |{\bf g},{\bf x} \right) \leq \exp \left(- \frac{a_0^2}{2({\mathcal D}^\eta({\bf x},{\bf g}) )_{ii}}\right) + \exp \left(- \frac{a_0^2}{2 ({\mathcal D}^\eta({\bf x},{\bf g}) )_{ii} + \frac{2}{3} a_0}  \right)
	$$
	where the absence of $\tau$ as a subscript indicates $\tau = 0$. Take $a_0 = a_1  ({\mathcal D}_\tau^\eta({\bf x},{\bf g}) )_{ii}$ where $0 < a_1 \leq 1$, then
	\begin{align*}
	\Pr\big[ | ({\mathcal D}_\tau^\eta({\bf x},{\bf g}) )_{ii} - (D_\tau^\eta)_{ii} | \geq a_1 ({\mathcal D}_\tau^\eta({\bf x},{\bf g}) )_{ii} |{\bf g},{\bf x}\big] &\leq 	2 \exp \left(- \frac{a_1^2 ({\mathcal D}_\tau^\eta({\bf x},{\bf g}) )_{ii} }{3} \right)\\
	& = 2 \exp \left( -\frac{a_1^2\left[ ({\mathcal D}^\eta({\bf x},{\bf g}) )_{ii} + \tau\right] }{3} \right) \\
	& \leq  2 \exp \left( -\frac{a_1^2\left[ d_{\min} + \tau\right] }{3} \right).
	\end{align*}
	Note that for any $x \geq 0$, $|\sqrt{x} - 1| \leq |x - 1|$.
	Thus
	we have 	
	\begin{align*}
			\Pr \left[\|I -   ({\mathcal D}^\eta_\tau({\bf x},{\bf g}) )^{-\frac 1 2}(D_\tau^\eta)^{\frac 1 2}\| \geq a_1|{\bf g},{\bf x}\right] &= \Pr \left[\max_{i \in \eta_N(x) \cup \eta_N(x')} \left|1 -  \sqrt{\frac{(D_\tau^\eta)_{ii}}{({\mathcal D}_\tau^\eta({\bf x},{\bf g}) )_{ii}}}\right| \geq a_1|{\bf g},{\bf x}\right]\\
			& \leq \Pr \left[\max_{i \in \eta_N(x) \cup \eta_N(x')} \left|1 -  \frac{(D_\tau^\eta)_{ii}}{({\mathcal D}_\tau^\eta({\bf x},{\bf g}) )_{ii}}\right| \geq a_1|{\bf g},{\bf x}\right]\\
			& \leq \sum_{i \in \eta_N(x) \cup \eta_N(x')} \Pr \left[| ({\mathcal D}_\tau^\eta({\bf x},{\bf g}) )_{ii} - (D_\tau^\eta)_{ii} | \geq a_1  ({\mathcal D}_\tau^\eta({\bf x},{\bf g}) )_{ii}|{\bf g},{\bf x}\right]\\
			& \leq  4k \exp \left( -\frac{a_1^2\left[d_{\min} + \tau\right] }{3} \right)
	\end{align*}
	On the event $\{ \|I -   ({\mathcal D}^\eta_\tau({\bf x},{\bf g}) )^{-\frac 1 2}(D_\tau^\eta)^{\frac 1 2}\| \leq a_1\}$, 
	$$
	\|({\mathcal D}^\eta_\tau({\bf x},{\bf g}) )^{-\frac 1 2} (D_\tau^\eta)^{\frac 1 2} \|    \leq \|I\| + \|I  -   (D_\tau^\eta)^{\frac 1 2}({\mathcal D}^\eta_\tau({\bf x},{\bf g}) )^{-\frac 1 2}   \| \leq 1 + a_1,
	$$
	which implies that $
	\|B_1\| + \|B_2\| \leq a_1^2 + 2 a_1 \leq 3a_1
	$.
	Thus, 
	\begin{align}
		\Pr\left(\|B_1\| + \|B_2\|  \leq 3a_1|{\bf g},{\bf x}\right) \geq 1- 4k \exp \left( -\frac{a_1^2\left[ d_{\min} + \tau\right] }{3} \right).\label{eq:boundB1B2}
	\end{align}
	The expression for $B_3$ simplifies to
	$$
	B_3 = ({\mathcal D}^\eta_\tau({\bf x},{\bf g}) )^{-\frac 1 2}\left[\widetilde{A^\eta} - \widetilde{P^\eta}({\bf x},{\bf g}) \right]  ({\mathcal D}^\eta_\tau({\bf x},{\bf g}) )^{-\frac 1 2}.
	$$
	We follow \citeasnoun{rohe2016co} and 
	decompose 
	$
		B_3  = \sum_{i \in \eta_N(x)} \sum_{j \in \eta_N(x')} Y_{i,j}
	$
	with 
	$$
	Y_{i,j} = \frac{A_{ij} - P_{ij}({\bf x},{\bf g}) }{\left( [{\mathcal O}_{ii}^\eta({\bf x},{\bf g})  + \tau] [{\mathcal Q}_{ii}^\eta({\bf x},{\bf g})  + \tau]  \right)^{1/2}} E^{i,k + j}
	$$
	where $E^{i,j}$ is the square matrix of size $2k$ with  ones at position $(i,j)$ and at position $(j,i)$ and zeros everywhere else. We apply a concentration inequality for symmetric matrices, see Theorem 5.4.1 of \citeasnoun{vershynin2018high}. 
	$\|E^{i,j} \| =1$ thus a bound on $\|Y_{i,j}\|$ is
	\begin{align*}
		\|Y_{i,j}\| & \leq  \left( [{\mathcal O}_{ii}^\eta({\bf x},{\bf g})  + \tau] [{\mathcal Q}_{ii}^\eta({\bf x},{\bf g})  + \tau]   \right)^{-1/2} \\ 
		& \leq \left(\left[d_{\min} + \tau\right]  \left[d_{\min} + \tau\right]  \right)^{-1/2} = \frac{1}{d_{\min} + \tau}
	\end{align*}
	and a bound on $\| \sum_{i \in \eta_N(x)} \sum_{j \in \eta_N(x')} \E[ (Y_{i,j})^2|{\bf g},{\bf x}] \|$
	is
	\begin{align*}
		\| \sum_{i \in \eta_N(x)}& \sum_{j \in \eta_N(x')}  \E[(Y_{i,j})^2|{\bf g},{\bf x}] \| = \| \sum_{i \in \eta_N(x)} \sum_{j \in \eta_N(x')}  \left[ \frac{ P_{ij}({\bf x},{\bf g}) - P_{ij}({\bf x},{\bf g})^2}{ [{\mathcal O}_{ii}^\eta({\bf x},{\bf g})  + \tau] [{\mathcal Q}_{jj}^\eta({\bf x},{\bf g})  + \tau] }  (E^{i,i} + E^{k+j, k+j})  \right] \|\\
		 &=  \|	\sum_{i \in \eta_N(x)} \left[  \sum_{j \in \eta_N(x')} \frac{ P_{ij}({\bf x},{\bf g}) - P_{ij}({\bf x},{\bf g})^2}{ [{\mathcal O}_{ii}^\eta({\bf x},{\bf g})  + \tau] [{\mathcal Q}_{jj}^\eta({\bf x},{\bf g})  + \tau] } \right] E^{i,i}\\
		 & \qquad \qquad \qquad + \sum_{j \in \eta_N(x')} \left[  \sum_{i \in \eta_N(x)} \frac{ P_{ij}({\bf x},{\bf g}) - P_{ij}({\bf x},{\bf g})^2}{ [{\mathcal O}_{ii}^\eta({\bf x},{\bf g})  + \tau] [{\mathcal Q}_{jj}^\eta({\bf x},{\bf g})  + \tau] }    \right]  E^{k+j, k+j} \| \\
		 & = \max\left(  \max_{i \in \eta_N(x)} \left( \sum_{j \in \eta_N(x')}  \frac{ P_{ij}({\bf x},{\bf g}) - P_{ij}({\bf x},{\bf g})^2}{ [{\mathcal O}_{ii}^\eta({\bf x},{\bf g})  + \tau] [{\mathcal Q}_{jj}^\eta({\bf x},{\bf g})  + \tau] } \right) , \right. \\
		 &\qquad \qquad \qquad \left. \max_{j \in \eta_N(x')} \left(  \sum_{i \in \eta_N(x)} \frac{ P_{ij}({\bf x},{\bf g}) - P_{ij}({\bf x},{\bf g})^2}{ [{\mathcal O}_{ii}^\eta({\bf x},{\bf g})  + \tau] [{\mathcal Q}_{jj}^\eta({\bf x},{\bf g})  + \tau] }  \right)  \right) \\
		 & \leq\max   \left( \frac{1}{d_{\min} + \tau} \max_{i \in \eta_N(x)} \left( \sum_{j \in \eta_N(x')}  \frac{ P_{ij}({\bf x},{\bf g})}{ {\mathcal O}_{ii}^\eta({\bf x},{\bf g})  + \tau  } \right) ,\frac{1}{d_{\min} + \tau}\max_{j \in \eta_N(x')} \left(  \sum_{i \in \eta_N(x)} \frac{ P_{ij}({\bf x},{\bf g}) }{ {\mathcal Q}_{jj}^\eta({\bf x},{\bf g})  + \tau }  \right)  \right) \\
		 & \leq \frac{1}{d_{\min} + \tau},
	\end{align*}
	where the third equality holds by definition of the spectral norm. We apply  Theorem 5.4.1 of \citeasnoun{vershynin2018high} and obtain 
	\begin{equation}\label{eq:boundB3}
		\Pr\left( \|B_3\|  \leq a_1|{\bf g},{\bf x} \right) \geq 1 - 4k\exp \left(-\frac{a_1^2 [d_{\min} + \tau] }{2+ 2a_1/3}\right).
	\end{equation}		
	The remaining term is $B_4$. We decompose
	\begin{eqnarray*}
		B_4 &=&
		\underbrace{
			[({\mathcal D}^\eta_\tau({\bf x},{\bf g}) )^{-\frac 1 2} - ({\mathcal D}^\eta_\tau(x,x',{\bf g}) )^{-\frac 1 2}] \widetilde{P^\eta}({\bf x},{\bf g})   ({\mathcal D}^\eta_\tau({\bf x},{\bf g}) )^{-\frac 1 2}  
		}_{= B_{41}} 
			\\
			&&+
			\underbrace{
		({\mathcal D}^\eta_\tau(x,x',{\bf g}) )^{-\frac 1 2} \widetilde{P^\eta}({\bf x},{\bf g})  	[({\mathcal D}^\eta_\tau({\bf x},{\bf g}) )^{-\frac 1 2} - ({\mathcal D}^\eta_\tau(x,x',{\bf g}) )^{-\frac 1 2}]
	}_{=B_{42}}
	\\	 &&+ 
		\underbrace{	
		({\mathcal D}^\eta_\tau(x,x',{\bf g}) )^{-\frac 1 2} \widetilde{P^\eta}({\bf x},{\bf g}) ({\mathcal D}^\eta_\tau(x,x',{\bf g}) )^{-\frac 1 2} 	
				- ({\mathcal D}^\eta_\tau(x,x',{\bf g}) )^{-\frac 1 2} \widetilde{P^\eta}(x,x',{\bf g})   ({\mathcal D}^\eta_\tau(x,x',{\bf g}) )^{-\frac 1 2}
			}_{=B_{43}}.
	\end{eqnarray*}
	 By Lemma \ref{lem:normmathcalL}, $\|\widetilde{\mathcal L_\tau^\eta}({\bf x},{\bf g}) \|  \leq 1$.  Thus, 
	\begin{eqnarray*}
	\|B_{41}\| &\leq& \|I - ({\mathcal D}^\eta_\tau(x,x',{\bf g}) )^{-\frac 1 2} ({\mathcal D}^\eta_\tau({\bf x},{\bf g}) )^{\frac 1 2}   \| \|\widetilde{\mathcal L_\tau^\eta}({\bf x},{\bf g}) \| 
	\\
	&\leq&  \|I - ({\mathcal D}^\eta_\tau(x,x',{\bf g}) )^{-\frac 1 2}  ({\mathcal D}^\eta_\tau({\bf x},{\bf g}) )^{\frac 1 2} \|,
	\end{eqnarray*}
	and
	\begin{eqnarray*}
		\|B_{42}\| &\leq& \| ({\mathcal D}^\eta_\tau(x,x',{\bf g}) )^{-\frac 1 2}  ({\mathcal D}^\eta_\tau({\bf x},{\bf g}) )^{\frac 1 2}\|  \|\widetilde{\mathcal L_\tau^\eta}({\bf x},{\bf g}) \| \|I -  ({\mathcal D}^\eta_\tau({\bf x},{\bf g}) )^{\frac 1 2}  ({\mathcal D}^\eta_\tau(x,x',{\bf g}) )^{-\frac 1 2} \|
		\\
		&\leq&  \| ({\mathcal D}^\eta_\tau(x,x',{\bf g}) )^{-\frac 1 2}  ({\mathcal D}^\eta_\tau({\bf x},{\bf g}) )^{\frac 1 2}\|   \|I - ({\mathcal D}^\eta_\tau({\bf x},{\bf g}) )^{\frac 1 2} ({\mathcal D}^\eta_\tau(x,x',{\bf g}) )^{-\frac 1 2}  \|.
	\end{eqnarray*}
	Note that
	\begin{align*}
		  \|I - ({\mathcal D}^\eta_\tau(x,x',{\bf g}) )^{-\frac 1 2}  ({\mathcal D}^\eta_\tau({\bf x},{\bf g}) )^{\frac 1 2} \| = \max_{i \in \eta_N(x) \cup \eta_N(x')} \left|1 -  \sqrt{\frac{({\mathcal D}^\eta_\tau({\bf x},{\bf g}) )_{ii}}{({\mathcal D}^\eta_\tau(x,x',{\bf g}) )_{ii}}}\right| \leq \max_{i \in \eta_N(x) \cup \eta_N(x')} \left|1 -  \frac{({\mathcal D}^\eta_\tau({\bf x},{\bf g}) )_{ii}}{({\mathcal D}^\eta_\tau(x,x',{\bf g}) )_{ii}}\right|
	\end{align*}
	If $i \in \eta_N(x) $, 
	\begin{align*}
		({\mathcal D}^\eta_\tau({\bf x},{\bf g}) )_{ii} = ({\mathcal O}_\tau^\eta({\bf x},{\bf g}))_{ii} = \sum_{j \in \eta_N(x')   } {\mathbb E}[A_{ij}^\eta|{\bf x},{\bf g}] + \tau = \sum_{j \in \eta_N(x')   } B_{g(i),g(j)}(x(i),x(j)) + \tau 
	\end{align*}
	thus 
	\begin{align*}
		\left| ({\mathcal D}^\eta_\tau({\bf x},{\bf g}) )_{ii} - ({\mathcal D}^\eta_\tau(x,x',{\bf g}) )_{ii} \right| &\leq \sum_{j \in \eta_N(x')   } | B_{g(i),g(j)}(x(i),x(j)) - B_{g(i),g(j)}(x,x')| \\
		& \leq \sum_{j \in \eta_N(x')} l_B \|(x(i),x(j)) - (x,x') \| \\
		&\leq  \sum_{j \in \eta_N(x')} l_B \left[||(x(i),x(j)) - (x(i),x')|| + || (x(i),x') - (x,x')|| \right] \\
		& = \sum_{j \in \eta_N(x')}  l_B \left[||x(j)- x'|| + || x(i) - x|| \right] \leq k l_B  (r_k(x) +  r_{k}(x') )\\
		& \leq 2 k l_B R_k,
	\end{align*}
	where the second inequality holds by Assumption \ref{assn:LipshitzB} and the last inequality by Condition (\ref{cond:radius}). We also have for $i \in \eta_N(x)$,
	\begin{equation}\label{eq:boundD}
		({\mathcal D}^\eta_\tau(x,x',{\bf g}) )_{ii} = {\mathcal O}_{ii}(x,x',{\bf g}) + \tau \geq d_{\min} + \tau
	\end{equation}
	which implies
	$$ 
	\left|1 -  \frac{({\mathcal D}^\eta_\tau({\bf x},{\bf g}) )_{ii}}{({\mathcal D}^\eta_\tau(x,x',{\bf g}) )_{ii}}\right| 
	\leq 
	\frac{  2 k l_B R_k}{d_{\min} + \tau}.
	$$
	The same bound holds for $i \in \eta_N(x')$. Thus
	\begin{equation}\label{eq:boundB41}
		\|B_{41}\|  \leq \frac{ 2 k l_B R_k}{ d_{\min} + \tau}.
	\end{equation}
	 We also obtain
	\begin{equation}\label{eq:boundB42}
		\|B_{42}\| \leq  \frac{2 k l_B R_k}{ d_{\min} + \tau} \left( \frac{2 k l_B R_k}{ d_{\min} + \tau} + 1\right).
	\end{equation}	 
	The spectral norm  of $B_{43}$ can be bounded as follows
	\begin{align*}
	  	\| B_{43}\| & \leq  \| {\mathcal D}^\eta_\tau(x,x',{\bf g}) \|^{-1} \, \| \widetilde{P^\eta}({\bf x},{\bf g}) - \widetilde{P^\eta}(x,x',{\bf g}) \|.
	\end{align*}
	By Equation \eqref{eq:boundD}, we have $\| {\mathcal D}^\eta_\tau(x,x',{\bf g}) \| \geq d_{\min} + \tau$. As for the second term in the inequality above,
	\begin{align*}
		\| \widetilde{P^\eta}({\bf x},{\bf g}) - \widetilde{P^\eta}(x,x',{\bf g}) \|  & \leq   k  \max_{\substack{i \in \eta_N(x) \\ j \in  \eta_N(x')}} \left(\widetilde{P_{ij}^\eta}({\bf x},{\bf g}) -  \widetilde{P_{ij}^\eta}(x,x',{\bf g})\right) \\
		& =   k  \max_{\substack{i \in \eta_N(x) \\ j \in  \eta_N(x')}} \left(B_{g(i),g(j)}(x(i),x(j)) - B_{g(i),g(j)}(x,x')\right) \\
		& \leq 2 k l_B R_k,
	\end{align*}
	where the first inequality comes from the fact that for any matrix $A$ of size $k \times k$, $\|A\| \leq \|A\|_{F} \leq  k \max_{i,j}|A_{ij}|$ and the following equality from Assumption \ref{assn:LipshitzB}. Thus 
	\begin{equation}\label{eq:boundB43}
		\|B_{43}\| \leq \frac{2 k l_B R_k}{d_{\min} + \tau}.
	\end{equation}
	Adding \eqref{eq:boundB41}, \eqref{eq:boundB42} and \eqref{eq:boundB43}, we obtain 
	\begin{align}
		\|B_{4}\| &\leq  \frac{2 k l_B R_k}{ d_{\min} + \tau} \left( \frac{2 k l_B R_k}{ d_{\min} + \tau} + 3\right).\label{eq:boundB4} 
	\end{align}
	
	Putting \eqref{eq:boundB1B2}, \eqref{eq:boundB3} and \eqref{eq:boundB4} together, and using the fact that for $L$ events $V_1$,...,$V_L$, $\Pr(\cap_{l=1}^L V_l) \geq \sum_{l=1}^L \Pr (V_l) - L +1$, we obtain
	\begin{align*}
		\Pr\left( \|\widetilde{L_\tau^\eta} - \widetilde{{\mathcal L}_\tau^\eta}(x,x',{\bf g}) \|  \right. & \left. \leq  4a_1 + \frac{2 k l_B R_k}{ d_{\min} + \tau} \left( \frac{2 k l_B R_k}{ d_{\min} + \tau} + 3\right) \ | \ {\bf g},{\bf x} \right) \\
		& \geq \Pr\left(\|B_1\| + \|B_2\|  \leq 3a_1|{\bf g},V({\bf g}) \right) +  \Pr\left( \|B_3\|  \leq a_1|{\bf g},{\bf x} \right)\\
		&\qquad  +  \Pr\left( \|B_4\|  \leq  \frac{2 k l_B R_k}{ d_{\min} + \tau} \left( \frac{2 k l_B R_k}{ d_{\min} + \tau} + 3\right)|{\bf g},{\bf x} \right) -2 \\
		&\geq 1- 8k \exp \left( -\frac{a_1^2\left[d_{\min} + \tau\right] }{3} \right).
	\end{align*}
	Taking $$a_1 = \sqrt{\frac{3 \ln (8k/\delta)}{d_{\min} + \tau}},$$
	then $a_1 \leq 1$ by Condition (\ref{cond:mindeg}), and
	\begin{align*}
		\Pr&\left( \|\widetilde{L_\tau^\eta} - \widetilde{{\mathcal L}_\tau^\eta}(x,x',{\bf g}) \|  \leq  4\sqrt{\frac{3 \ln (8k/\delta)}{d_{\min} + \tau}} + \frac{2 k l_B R_k}{ d_{\min} + \tau} \left( \frac{2 k l_B R_k}{ d_{\min} + \tau} + 3\right)  \ | \ {\bf g}, {\bf x} \right) \geq 1 - \delta.
	\end{align*}	
\end{proof}

\subsection{Lower bound on $d_{\min}(x,x', {\bold x})$}

\

The finite sample bound obtained in Lemma \ref{lem:boundLaplacians} is valid under restrictions on $d_{\min}(x,x', {\bold x})$, a localized minimum expected degree. In this section, we derive a nonasymptotic probability bound on $d_{\min}(x,x', {\bold x})$   and combine it with Lemmas \ref{lem:boundLaplacians} and \ref{lem:radius} to obtain a finite sample bound on $ \|\widetilde{L_\tau^\eta} - \widetilde{{\mathcal L}_\tau^\eta}(x,x',{\bf g}) \| $ conditional on ${\bf g}$ only and with explicit dependence in $k$. For $h \in [G]$, let
\begin{align*}
	n_h(x)& :=\# \{g(i)=h, i \in \eta_N(x)\}\\
	N_h & :=\# \{g(i)=h, i \in [N] \}\\
	\underline{n_h} &:= \min_{x \in S}n_h(x).
\end{align*}
Define 
$$
\Delta =
\inf_{(x,x') \in S^{2}} \min_{g \in [G]} \max_{h \in [G]} B_{gh}(x,x')
$$
Note that as for $B$ and $l_B$, $\Delta$  may vary with $N$ and thus captures sparsity. We assume that $\Delta >0$. A natural lower bound on $d_{\min}(x,x', {\bold x})$ therefore involves $\min_{h \in [G]} \underline{n_h}$. We show in Lemma \ref{lem:radius} that under certain assumptions, $\inf_{x \in S} r_k(x) \geq \underline{R}_k$ holds with high probability, where 
\begin{equation}
	\underline{R}_{k} := \left(\frac{ k - 12 d \ln(12N/\delta) }{4 N \overline{U}_X V_d}\right)^{1/d}. \label{eq:minradius}
\end{equation}
The following lemma uses this result and establishes a lower bound on $\min_{h \in [G]} \underline{n_h}$.

\begin{lem}\label{lem:bdlocalgpsize}
	Let Assumptions \ref{assn:boundsupport}, \ref{assn:boundinff} and \ref{assn:boundsupf} hold. Then for all $N$, $\delta \in (0,1)$, $1 \leq k \leq N$ such that
	\begin{enumerate}
		\item \label{cond:lowbd}$k \geq 12 d \ln(24GN/\delta)$,
		\item \label{cond:upperbd} $k \leq 8  T^d  V_d \overline{U}_X N$, 
		\item \label{cond:floorlowbd} and for all $h \in [G],$ $\frac{c}{16} \frac{N_h}{N} \frac{b_X}{\overline{U}_X} \ k \geq 24 d \ln(24 G N_h/\delta)  + 1$,
	\end{enumerate}
	with probability at least $1-\delta$  conditional on ${\bold g}$, it holds that
	$$ \min_{h \in [G]} \underline{n_h} \geq \min_{h \in [G]} \left \lfloor{\frac{c}{16} \frac{N_h}{N} \frac{b_X}{\overline{U}_X} \ k}\right \rfloor.$$
\end{lem}

\begin{proof}
	Fix $\delta \in (0,1)$ and $k$ satisfying the conditions of the lemma.
	Note that 
	$$
	\underline{R}_k \geq \left(\frac{ k}{8 N \overline{U}_X V_d}\right)^{1/d}.
	$$
	We apply Lemma 4 of \citeasnoun{portier2021nearest} to the subpopulation $\{i \, : \, g(i) = h\}$. We define $r_l^h(x)$ to be the $l$-NN radius of $x \in S$ for this subpopulation, that is,
	$$r_l^h(x) := \inf \{ r>0 \, : \, |\mathcal{B}(x,r) \cap \{x(i) \, : \, g(i) =h \} | = l \}.$$
	 By Assumptions \ref{assn:boundinff} and \ref{assn:boundsupf}, we have
	\begin{equation}
		0 < b_X \leq f(x|h) \leq \overline{U}_X, \ \forall \, x \in S,
	\end{equation}
	and by Assumption \ref{assn:boundsupport}, Condition (7) of \citeasnoun{portier2021nearest} also holds for $f(x|h)$. 
	For any $l$ and $\delta'>0$ such that $24 d \ln(12N_h/\delta') \leq l \leq T^d N_h b_X c V_d/2$, let
	$$
	\overline{\tau}_{l}^{h} := \left(\frac{2l}{N_h b_X c V_d}\right)^{1/d}.
	$$
	By Lemma 4 of \citeasnoun{portier2021nearest}, with probability at least $1 - \delta'$  conditional on ${\bold g}$, it holds that
	$$\sup_{x \in S} r_l^h(x) \leq \overline{\tau}_{l}^{S,h}.$$
	We take $l_h = \left \lfloor{\frac{c}{16} \frac{N_h}{N} \frac{b_X}{\overline{U}_X} \ k}\right \rfloor $ and $\delta' = \frac{\delta}{2G}$. Then 
	$$l_h \leq \frac{c}{16} \frac{N_h}{N} \frac{b_X}{\overline{U}_X} \ k \leq  T^d N_h b_X c V_d/2,$$
	by Condition (\ref{cond:upperbd}) and 
	$$l_h \geq 24 d \ln(12N_h/\delta'),$$ 
	by Condition (\ref{cond:floorlowbd}). Moreover,
	\begin{align*}
		\overline{\tau}_{l_h}^{h} \leq  \left(\frac{k}{8 N \overline{U}_X V_d}\right)^{1/d} \leq \underline{R}_k.
	\end{align*}
	This implies
	\begin{align*}
		\{ \underline{n_h} \geq l_h\} = \{ n_h(x) \geq l_h \text{ for every } x \in S\} &= \{r_{l_h}^h(x) \leq r_k(x) \text{ for every } x \in S\ \}
		\\
		&\supseteq \{\sup_{x \in S}r_{l_h}^h(x) \leq \overline{\tau}_{l_h}^{h}\} \cap \{\inf_{x \in S}r_{k}(x) \geq \underline{R}_k \}.
	\end{align*}
	Define $L = \min_{h \in [G]} \left \lfloor{\frac{c}{16} \frac{N_h}{N} \frac{b_X}{\overline{U}_X} \ k}\right \rfloor = \min_{h \in [G]} l_h$. We apply the equation above to all $h \in [G]$  and obtain 
	\begin{align*}
		& \Pr\left( \min_{h \in [G]} \underline{n_h} \geq L \, | \, {\bold g}\right) \geq \sum_{h \in [G]} \Pr\left( \underline{n_h} \geq L \, | \, {\bold g}\right)- G + 1\\
		& \geq \sum_{h \in [G]} \Pr\left( \underline{n_h} \geq l_h \, | \, {\bold g}\right)- G + 1\\
		&\geq \sum_{h \in [G]}\Pr\left( \{\sup_{x \in S}r_{l_h}^h(x) \leq \overline{\tau}_{l_h}^{h}\} \cap \{\inf_{x \in S}r_{k}(x) \geq \underline{R}_k  \} \, | \, {\bold g}\right)- G + 1\\
		&\geq \sum_{h \in [G]} \left[ \Pr\left( \sup_{x \in S}r_{l_h}^h(x) \leq \overline{\tau}_{l_h}^{h}\, | \, {\bold g}\right) +  \Pr\left(\inf _{x \in S} r_{k }(x) \geq \underline{R}_k \, | \, {\bold g}\right) -1 \right] - G + 1\\
		&\geq G(1 - 2\delta') - G + 1 = 1 -  \delta.
	\end{align*}
	where in the last inequality, we used  $k \geq 12 d \ln(12N/\delta')$, guaranteed by Condition (\ref{cond:lowbd}), together with \eqref{eq:rad_upperbound} of Lemma \ref{lem:radius}. 
\end{proof}

\begin{lem}\label{lem:bdlocaldegree} 
	Let Assumptions \ref{assn:boundsupport}, \ref{assn:boundinff} and \ref{assn:boundsupf} hold. Then for all $N$, $\delta \in (0,1)$, $1 \leq k \leq N$ such that
	\begin{enumerate}
		\item $k \geq 12 d \ln(24GN/\delta)$,
		\item $k \leq 8  T^d  V_d \overline{U}_X N$, 
		\item and for all $h \in [G],$ $\frac{c}{16} \frac{N_h}{N} \frac{b_X}{\overline{U}_X} \ k \geq 24 d \ln(24 G N_h/\delta)  + 1$,
	\end{enumerate}
	with probability at least $1-\delta$  conditional on ${\bold g}$, it holds that
	$$d_{\min}(x,x', {\bold x}) \geq \Delta \min_{h \in [G]} \left \lfloor{\frac{c}{16} \frac{N_h}{N} \frac{b_X}{\overline{U}_X} \ k}\right \rfloor.$$
\end{lem}
\begin{proof}
	Recall that $$d_{\min}(x,x', {\bold x}) = \min \left(\min_{i \in \eta_N(x)} {\mathcal O}_{ii}(x,x',{\bf g}), \min_{j \in \eta_N(x')} {\mathcal Q}_{jj}(x,x',{\bf g}),\min_{i \in \eta_N(x)} {\mathcal O}_{ii}({\bf x},{\bf g}), \min_{j \in \eta_N(x')} {\mathcal Q}_{jj}({\bf x},{\bf g})\right).$$
	Note that 
	\begin{align*}
		{\mathcal O}_{ii}(x,x',{\bf g}) &= \sum_{j \in \eta_N(x')   }   P_{ij}(x,x',{\bf g}) = \sum_{h \in [G]} n_h(x') B_{g(i) h} (x,x')\\
		&\geq  \min_{h \in [G]}n_h(x') \max_{h' \in [G]} B_{g(i)h'}(x,x')\\
		&\geq  \min_{h \in [G]}\min_{x \in S}n_h(x) \inf_{(x,x') \in S^{2}} \min_{g \in [G]}\max_{h' \in [G]} B_{gh'}(x,x')\\
		&\geq \Delta \min_{h \in [G]} \underline{n_h}.
	\end{align*}
	The same inequality holds for ${\mathcal Q}_{jj}(x,x',{\bf g})$. For the remaining terms in the definitions of $d_{\min}(x,x', {\bold x})$, note that
	\begin{align*}
		{\mathcal O}_{ii}({\bf x},{\bf g}) &= \sum_{j \in \eta_N(x') }  P_{ij}({\bf x},{\bf g})
		\\
		&=\sum_{j \in \eta_N(x') }  {\mathbb E}[A_{ij}^\eta|x(i),x(j),g(i),g(j)] 
		\\
		&=\sum_{h \in [G]}  \sum_{j \in \eta_N(x'): g(j) = h}  B_{g(i) h}(x(i),x(j))
		\\
		&\geq \min_{h \in [G]}n_h(x')  \inf_{(x,x') \in S^{2}}   \min_{g \in [G]}\max_{h' \in [G]} B_{gh'}(x,x')
		\\
		&\geq \Delta \min_{h \in [G]} \underline{n_h}
	\end{align*}	
	This also holds for ${\mathcal Q}_{jj}({\bf x},{\bf g})$. Thus $d_{\min}(x,x', {\bold x}) \geq \Delta \min_{h \in [G]} \underline{n_h}$ and the result holds by Lemma \ref{lem:bdlocalgpsize}.
\end{proof}

Combining Lemma \ref{lem:boundLaplacians} together with Lemma \ref{lem:bdlocaldegree}, we obtain a new bound on $\|\widetilde{L^\eta} - \widetilde{{\mathcal L}^\eta}(x,x',{\bf g}) \|$.

\begin{lem}\label{lem:integr}
	Let Assumptions \ref{assn:boundsupport}, \ref{assn:boundinff}, \ref{assn:boundsupf} and \ref{assn:LipshitzB} hold. Then for all $N$, $\delta \in (0,1)$, $1 \leq k \leq N$ such that
	\begin{enumerate}
		\item $\Delta \min_{h \in [G]} \left \lfloor{\frac{c}{16} \frac{N_h}{N} \frac{b_X}{\overline{U}_X} \ k}\right \rfloor + \tau \geq 3 \ln (24k/\delta)$,
		\item $k \geq \max(12 d \ln(72GN/\delta),24 d \ln(36N/\delta))$,
		\item $k \leq \min(8  T^d  V_d \overline{U}_X N, (1/2)T^d  V_d b_X c N )$,
		\item and for all $h \in [G],$ $\frac{c}{16} \frac{N_h}{N} \frac{b_X}{\overline{U}_X} \ k \geq 24 d \ln(72 G N_h/\delta)  + 1$,
	\end{enumerate}
	with probability at least $1-\delta$ conditional on ${\bold g}$, it holds that
	\begin{align}
		 \|\widetilde{L_\tau^\eta} - \widetilde{{\mathcal L}_\tau^\eta}(x,x',{\bf g}) \|  \leq \, &  4\sqrt{\frac{3 \ln (24k/\delta)}{\Delta \min_{h \in [G]} \left \lfloor{\frac{c}{16} \frac{N_h}{N} \frac{b_X}{\overline{U}_X} \ k}\right \rfloor + \tau}} \nonumber \\
		 & + \frac{2 k l_B R_k}{ \Delta \min_{h \in [G]} \left \lfloor{\frac{c}{16} \frac{N_h}{N} \frac{b_X}{\overline{U}_X} \ k}\right \rfloor + \tau} \left( \frac{2 k l_B R_k}{ \Delta \min_{h \in [G]} \left \lfloor{\frac{c}{16} \frac{N_h}{N} \frac{b_X}{\overline{U}_X} \ k}\right \rfloor+ \tau} + 3\right). \label{eq:Lrate}
	\end{align}
\end{lem}

\begin{proof}
	In what follows we use some simplified notation. Let $C$ be the function such that $C(d_{\min})$ is the right hand side of \eqref{eq:condLapl}. We also define the following two events, $V({\bf g},{\bf x}):=\left\lbrace \sup_{x \in S} r_k(x) \leq R_k \right\rbrace$ and $W({\bf g},{\bf x}):= \left\lbrace \min_{h \in [G]} \underline{n_h} \geq \min_{h \in [G]} \left \lfloor{\frac{c}{16} \frac{N_h}{N} \frac{b_X}{\overline{U}_X} \ k}\right \rfloor \right \rbrace \subset
	\left\lbrace  d_{\min} \geq \Delta \min_{h \in [G]} \left \lfloor{\frac{c}{16} \frac{N_h}{N} \frac{b_X}{\overline{U}_X} \ k}\right \rfloor \right\rbrace$ where the inclusion follows by the proof of Lemma \ref{lem:bdlocaldegree}. Then,
	\begin{align*}
		\Pr &\left[  \|\widetilde{L_\tau^\eta} - \widetilde{{\mathcal L}_\tau^\eta}(x,x',{\bf g}) \|  \leq C\left(\Delta \min_{h \in [G]} \left \lfloor{\frac{c}{16} \frac{N_h}{N} \frac{b_X}{\overline{U}_X} \ k}\right \rfloor  \right) \,  \Big| \, {\bf g } \right] \\
		& \geq \Pr \left[  \left\lbrace \|\widetilde{L_\tau^\eta} - \widetilde{{\mathcal L}_\tau^\eta}(x,x',{\bf g}) \|  \leq C\left(\Delta \min_{h \in [G]} \left \lfloor{\frac{c}{16} \frac{N_h}{N} \frac{b_X}{\overline{U}_X} \ k}\right \rfloor  \right) \right\rbrace \cap V({\bf g},{\bf x}) \cap W({\bf g},{\bf x}) \,  \Big| \, {\bf g } \right]\\
		& \geq \Pr \left[  \left\lbrace \|\widetilde{L_\tau^\eta} - \widetilde{{\mathcal L}_\tau^\eta}(x,x',{\bf g}) \|  \leq C\left(d_{\min} \right) \right\rbrace \cap V({\bf g},{\bf x}) \cap W({\bf g},{\bf x}) \,  \Big| \, {\bf g } \right]\\
		& =\E\left[\Pr \left[  \left\lbrace \|\widetilde{L_\tau^\eta} - \widetilde{{\mathcal L}_\tau^\eta}(x,x',{\bf g}) \|  \leq C\left(d_{\min} \right) \right\rbrace  \,  \Big| \, {\bf g }, {\bf x}, V({\bf g},{\bf x}), W({\bf g},{\bf x}) \right]\,  \Big| \, {\bf g }, V({\bf g},{\bf x}), W({\bf g},{\bf x}) \right]\\
		&\quad \times \Pr\left[V({\bf g},{\bf x}) \cap W({\bf g},{\bf x}) \,  \Big| \, {\bf g }\right].
	\end{align*}
	By Lemma \ref{lem:boundLaplacians} and $\Delta \min_{h \in [G]} \left \lfloor{\frac{c}{16} \frac{N_h}{N} \frac{b_X}{\overline{U}_X} \ k}\right \rfloor + \tau \geq 3 \ln (24k/\delta)$,
	$$\Pr \left[  \left\lbrace \|\widetilde{L_\tau^\eta} - \widetilde{{\mathcal L}_\tau^\eta}(x,x',{\bf g}) \|  \leq C\left(d_{\min} \right) \right\rbrace  \,  \Big| \, {\bf g }, {\bf x}, V({\bf g},{\bf x}), W({\bf g},{\bf x}) \right] \geq 1 - \delta/3.$$
	Moreover, by Lemma \ref{lem:radius} and  $24 d \ln(36N/\delta) \leq k \leq T^d N b_X c V_d/2$,
	$$\Pr\left[V({\bf g},{\bf x}) \,  \Big| \, {\bf g }\right] \geq 1 - \delta/3.$$
	Since
	\begin{enumerate}
		\item $k \geq 12 d \ln(72GN/\delta)$,
		\item $k \leq 8  T^d  V_d \overline{U}_X N$, 
		\item and for all $h \in [G],$ $\frac{c}{16} \frac{N_h}{N} \frac{b_X}{\overline{U}_X} \ k \geq 24 d \ln(72 G N_h/\delta)  + 1$,
	\end{enumerate}
	we also have by Lemma \ref{lem:bdlocaldegree},
		$$ \Pr\left[ W({\bf g},{\bf x}) \,  \Big| \, {\bf g }\right] \geq 1 - \delta/3.$$
	Thus 
	\begin{align*}
		\Pr &\left[  \|\widetilde{L_\tau^\eta} - \widetilde{{\mathcal L}_\tau^\eta}(x,x',{\bf g}) \|  \leq C\left(\Delta \min_{h \in [G]} \left \lfloor{\frac{c}{16} \frac{N_h}{N} \frac{b_X}{\overline{U}_X} \ k}\right \rfloor  \right) \,  \Big| \, {\bf g } \right] \\
		& \geq (1 - \delta/3) \left(\Pr\left[V({\bf g},{\bf x}) \,  \Big| \, {\bf g }\right] + \Pr\left[ W({\bf g},{\bf x}) \,  \Big| \, {\bf g }\right]- 1 \right) \geq (1 - \delta/3)(1 - 2\delta/3) \geq 1 - \delta.
	\end{align*}
\end{proof}

\begin{rem}
	The variance term is the first term in \eqref{eq:Lrate}.  Consider the case where no regularization is made ($\tau = 0$) and  sparsity is captured by a parameter $\rho_N$, i.e., such that $B = \rho_N B_0$ where $B_0$ does not vary with $N$, see e.g. \citeasnoun{bickel2011method}). Approximating $N_h/N \approx C$ for all $h \in [G]$, the variance term  in \eqref{eq:Lrate} simplifies to 
	\begin{equation}\label{eq:varLapl}
		C \sqrt{\frac{\ln (24k/\delta)}{\rho_N k}},
	\end{equation}
	where $C$ is a constant subject to change in value. The normalization in the definition of the Laplacian is not unlike that of the  $k$-NN regression estimator. However in comparison to standard $k$-NN regression, e.g., \citeasnoun{jiang2019non}, the variance term has an extra $\sqrt{\ln k}$ term in the numerator, which comes from the growing dimension of $\widetilde{L_\tau^\eta}$ and appears when applying various matrix concentration inequalities. On the other hand, \eqref{eq:varLapl} does not have a $\ln N$ term in the numerator because it is  a pointwise bound at $(x,x') \in S^2$.
\end{rem}

\begin{rem}
	Let $\tau = 0$, sparsity be captured by a parameter $\rho_N$ and $N_h/N \approx C$ for all $h \in [G]$. Then the bias  term  in \eqref{eq:Lrate} simplifies to $\approx R_k(R_k + C) \approx R_k$ if $R_k \to 0$ as $N \to \infty$, where $C$ is a constant subject to change in value. This is similar to \citeasnoun{jiang2019non}, confirming the intuition that the effective dimension is that of the covariates and not twice as much. Note that the bias, a novel term in this type of calculations, is not impacted by sparsity.	
\end{rem}

\begin{rem}\label{rem:opt_k}
	Let $\tau = 0$, sparsity be captured by a parameter $\rho_N$ and $N_h/N \approx C$ for all $h \in [G]$. Then as long as $N \rho_N \to \infty$, taking $k \approx (N^2 / \rho_N^d)^{\frac{1}{d+2}}$ 
	gives $\|\widetilde{L^\eta} - \widetilde{{\mathcal L}^\eta}(x,x',{\bf g}) \| \approx (\rho_N N)^{\frac{-1}{d+2}}$, up to $\ln N$ and $\ln \rho_N$ factors: the rate obtained is as \citeasnoun{jiang2019non}, see Remark 1, where the sample size is replaced with $N \rho_N $ due to sparsity.
\end{rem}

\subsection{Clustering}

\

In this section, for ease of readability, we do not display dependence of many of the defined objects in $x,x',\tau, k$, etc. Using the Laplacian ${L_\tau^\eta}$ we compute its top $G$ left/right singular vectors, to obtain the SVD of  ${L_\tau^\eta}$
$$
{L_\tau^\eta}  = U  \widehat \Lambda V^\top
$$
where $U, V  \in {\mathbb R^{k \times G}}$.  We then apply $K$-means clustering to $U$ and $V$, solving 
\begin{align}
	(\widehat{\Theta}^\eta(x),Z_U) &= \argmin_{\Theta \in \mathbb{M}_{k,G}, Z \in \mathbb{R}^{G \times G}}\|\Theta Z - U\|_F 
	\\
	(\widehat{\Theta}^\eta(x'),Z_V) &= \argmin_{\Theta \in \mathbb{M}_{k,G}, Z \in \mathbb{R}^{G \times G}}\|\Theta Z - V\|_F 
\end{align}
It is possible to use a faster clustering algorithm by using an approximate solution instead, incorporating the approximation error explicitly in the following analysis as in  \citeasnoun{lei2015consistency}.    

Note that the $i$-th element of the diagonal matrix ${\mathcal O}_\tau^\eta(x,x',{\bf g})$ only depends on $g(i)$. Define the diagonal matrix $\overline{{\mathcal O}} \in \R^{G \times G}$ collecting these coefficients, i.e., such that 
$\overline{{\mathcal O}}_{gg} = \sum_{j \in \eta_N(x')} B_{gg(j)}(x,x') + \tau$. Similarly, define the diagonal matrix $\overline{{\mathcal Q}} \in \R^{G \times G}$ such that $\overline{{\mathcal Q}}_{gg} = \sum_{i \in \eta_N(x)} B_{g(i)g}(x,x') + \tau$.

Define $\mathcal{N}(x) = {\mathrm{diag}}\left({\sqrt{n_1(x)}},...,{\sqrt{n_G(x)}}\right)$ and 
let 
\begin{equation}\label{eq:defDcal}
	\mathcal{N}(x) \overline{{\mathcal O}}^{-1/2} B(x,x') \overline{{\mathcal Q}}^{-1/2} \mathcal{N}(x') = z_U \Lambda_\tau z_V^\top
\end{equation}
be the SVD of  $B(x,x')$ after normalization by $\mathcal{N}(x) \overline{{\mathcal O}}^{-1/2}$ and $\mathcal{N}(x') \overline{{\mathcal Q}}^{-1/2}$.
Note that
\begin{align*}
	{{\mathcal L}_\tau^\eta}(x,x',{\bf g})  &= ({\mathcal O}_\tau^\eta(x,x',{\bf g}))^{-\frac 1 2} \Theta^\eta(x) B(x,x') {\Theta^\eta(x')}^\top ({\mathcal Q}_\tau^\eta(x,x',{\bf g}))^{-\frac 1 2},
\end{align*}
with $({\mathcal O}_\tau^\eta(x,x',{\bf g}))^{-\frac 1 2} \Theta^\eta(x) = \Theta^\eta(x) \overline{{\mathcal O}}^{-1/2}$ and $({\mathcal Q}_\tau^\eta(x,x',{\bf g}))^{-\frac 1 2} \Theta^\eta(x') = \Theta^\eta(x') \overline{{\mathcal Q}}^{-1/2}$. Thus
\begin{align}\label{eq:decompPopLagr}
	{{\mathcal L}_\tau^\eta}(x,x',{\bf g})  
	&= \Theta^\eta(x) \mathcal{N}(x)^{-1}z_U \Lambda_\tau z_V^\top  \mathcal{N}(x')^{-1}  {\Theta^\eta(x')}^\top.
\end{align}
From this expression, by slight modification of the argument in the proof of Lemma 2.1 in  \citeasnoun{lei2015consistency}, considering SVD instead of eigenvalue decomposition,  and taking account of the normalization of the Laplacian, we see that 
the SVD of ${{\mathcal L}^\eta}(x,x',{\bf g})$ is given by 
$$
{{\mathcal L}^\eta}(x,x',{\bf g})  = {\mathcal U} \Lambda_\tau {\mathcal V^\top}
$$ 
with $\mathcal U =  \Theta^\eta(x) \mathcal Z_U$ and  $\mathcal V = \Theta^\eta(x') \mathcal Z_V$ where  $\mathcal Z_U = \mathcal{N}(x)^{-1} z_U$ and  $\mathcal Z_V = \mathcal{N}(x)^{-1} z_V$.  

Note 
\begin{equation}\label{eq:centroid}
	\mathcal Z_U \mathcal Z_U^\top = \mathcal{N}(x)^{-1}z_U z_U^\top \mathcal{N}(x)^{-1} = {\mathrm{diag}}\left(\frac 1 {n_1(x)},...,\frac 1 {n_G(x)}    \right)
\end{equation}

As in  \citeasnoun{rohe2016co}, define
$$
\tilde {\mathcal U}: = \frac{1}{\sqrt 2}\binom{\mathcal U}{\mathcal V},
$$ 
$$
\tilde {U}: = \frac{1}{\sqrt 2}\binom{U}{V},
$$ 
then $\tilde {\mathcal U}$ and $\tilde {U}$ are the eigenvectors corresponding to top $G$ eigenvalues  of $\widetilde{\mathcal L^\eta}$ and $\widetilde{L^\eta}$, respectively.  We will use $\lambda_1(x,x') \geq \lambda_2(x,x') \geq ... \geq \lambda_G(x,x')$ and  $\hat \lambda_1(x,x') \geq \hat \lambda_2(x,x') \geq ... \geq \hat \lambda_G(x,x')$ to denote them although we drop $(x,x')$ in most of our computations.
We now have
\begin{lem}\label{lem:davis}
	$$
	\|\tilde U \tilde U^\top - \tilde {\mathcal U}\tilde {\mathcal U}^\top  \|_F \leq \frac{2\sqrt{2G}}{\lambda_G   }\| \widetilde{L_\tau^\eta} -  \widetilde{\mathcal L_\tau^\eta}   \|.
	$$	
	Moreover, for some $G \times G$ orthogonal matrices $Q_U$ and $Q_V$, 
	$$
	\|U - \mathcal{U} Q_U \|_F \leq \frac{4\sqrt{2G}}{\lambda_G   }\| \widetilde{L_\tau^\eta} -  \widetilde{\mathcal L_\tau^\eta}   \|.
	$$
	and 
	$$
	\|V - \mathcal{V} Q_V \|_F \leq \frac{4\sqrt{2G}}{\lambda_G   }\| \widetilde{L_\tau^\eta} -  \widetilde{\mathcal L_\tau^\eta}   \|.
	$$
\end{lem}
\begin{proof}  We employ the proof strategy developed by \citeasnoun{lei2015consistency} and \citeasnoun{rohe2016co}  with suitable  modification.     
	Note that these eigenvalues are equal to the top $G$ singular values of  ${\mathcal L^\eta}$ and ${L^\eta}$, respectively.    By Davis-Kahan Theorem (see Theorem VII.3.1 in \citeasnoun{bhatia2013matrix}) and noting that the eigengap between $\hat \lambda_j, j \leq G$ and  $\lambda_j, j \geq G+1$ is $\hat \lambda_G$  we have    
	$$
	\|(I - \tilde {\mathcal U}\tilde {\mathcal U}^\top) \tilde U \tilde U^\top \| \leq \frac{1}{\hat \lambda_G}\| \widetilde{L_\tau^\eta} -  \widetilde{\mathcal L_\tau^\eta}   \|.
	$$
	But 
	\begin{align*}
		\hat \lambda_G & \geq  \lambda_G - |\hat \lambda_G - \lambda_G|
		\\
		&\geq \lambda_G - \| \widetilde{L_\tau^\eta} -  \widetilde{\mathcal L_\tau^\eta}   \|
	\end{align*}
	where the second inequality follows from Weyl's Theorem.  It follows that
	$$
	\|(I - \tilde {\mathcal U}\tilde {\mathcal U}^\top) \tilde U \tilde U^\top \| \leq \frac{1}{\lambda_G -\| \widetilde{L_\tau^\eta} -  \widetilde{\mathcal L_\tau^\eta}   \|}\| \widetilde{L_\tau^\eta} -  \widetilde{\mathcal L_\tau^\eta}   \|.
	$$
	If $\| \widetilde{L_\tau^\eta} -  \widetilde{\mathcal L_\tau^\eta}   \| \leq \frac{\lambda_G}{2}$ then we have
	$$
	\|(I - \tilde {\mathcal U}\tilde {\mathcal U}^\top) \tilde U \tilde U^\top \| \leq \frac{2}{\lambda_G   }\| \widetilde{L_\tau^\eta} -  \widetilde{\mathcal L_\tau^\eta}   \|.
	$$
	If, on the other hand, $\| \widetilde{L_\tau^\eta} -  \widetilde{\mathcal L_\tau^\eta}\| > \frac{\lambda_G}{2}  $ then we directly have
	$$
	\|(I - \tilde {\mathcal U}\tilde {\mathcal U}^\top) \tilde U \tilde U^\top \| \leq 1 \leq \frac{2}{\lambda_G   }\| \widetilde{L_\tau^\eta} -  \widetilde{\mathcal L_\tau^\eta}   \|
	$$      
	again.  
	Now, noting rank$(\hat U)$ is at most $G$ it holds that
	\begin{align*}
		\|(I - \tilde {\mathcal U}\tilde {\mathcal U}^\top) \tilde U \tilde U^\top \| & \geq \frac{1}{\sqrt{G}} \|(I - \tilde {\mathcal U}\tilde {\mathcal U}^\top) \tilde U \tilde U^\top \|_F
		\\
		& \geq  \frac{1}{\sqrt{2G}} \|\tilde U \tilde U^\top - \tilde {\mathcal U}\tilde {\mathcal U}^\top  \|_F.
	\end{align*}
	by Proposition 2.1 in \citeasnoun{vu2013minimax}. In sum, we have 
	$$
	\|\tilde U \tilde U^\top - \tilde {\mathcal U}\tilde {\mathcal U}^\top  \|_F \leq \frac{2\sqrt{2G}}{\lambda_G   }\| \widetilde{L_\tau^\eta} -  \widetilde{\mathcal L_\tau^\eta}   \|
	$$
	as desired.  For the second assertion, as in  \citeasnoun{rohe2016co} we note that for some $G \times G$ orthogonal matrix $Q_U$ 
	\begin{align*}
		\|\tilde U \tilde U^\top - \tilde {\mathcal U}\tilde {\mathcal U}^\top  \|_F &\geq \frac{1}{2}\|U U^\top -  {\mathcal U} {\mathcal U}^\top \|_F
		\\
		&\geq  \frac{1}{\sqrt{2}} \|\sin \Theta({\mathrm{col}(U),\mathrm{col}({\mathcal U})}) \|_F
		\\
		& \geq \frac{1}{2} \|U - \mathcal{U} Q_U \|_F
	\end{align*}
	where the second inequality follows from Proposition 2.1 in \citeasnoun{vu2013minimax} and the third follows from Proposition 2.2 in \citeasnoun{vu2013minimax}.   We now have 
	$$
	\|U - \mathcal{U} Q_U \|_F \leq \frac{4\sqrt{2G}}{\lambda_G   }\| \widetilde{L_\tau^\eta} -  \widetilde{\mathcal L_\tau^\eta}   \|.
	$$
	A similar argument shows the last inequality. 
\end{proof}

We now bound misclassification rates. Let 
$$
\bar U = \widehat{\Theta}^\eta(x)Z_U,\quad  \bar V = \widehat{\Theta}^\eta(x')Z_V,
$$
then
\begin{align*}
	\|\bar U - \mathcal U Q_U\|_F^2 &\leq 2 \|\bar U -  U\|_F^2 + 2 \|U - \mathcal U Q_U\|_F^2
	\\
	& \leq  2 \|\mathcal U Q_U -  U\|_F^2 + 2 \|U - \mathcal UQ_U\|_F^2 
	\\
	& = 4 \|U - \mathcal U Q_U\|_F^2
\end{align*}
and likewise
$$
\|\bar V - \mathcal VQ_V\|_F^2 \leq 4\|V - \mathcal V Q_V\|_F^2.  
$$
Define 
$$
\mathcal Z_U' = \mathcal Z_U Q_U, \quad  \mathcal Z_V' = \mathcal Z_V Q_V,
$$
Note that for $g \in [G]$, by \eqref{eq:centroid} 
\begin{align}\label{eq:z-diff}
	\min_{\ell \neq g}\| {\mathcal Z}'_{U_{\ell \cdot}}  -   {\mathcal Z}'_{U_{g \cdot}} \|
	& = \min_{\ell \neq g}\| {\mathcal Z}_{U_{\ell \cdot}}  -   {\mathcal Z}_{U_{g \cdot}} \| \nonumber
	\\
	&= \min_{\ell \neq g} \sqrt{\frac{1}{n_g(x)} +  \frac{1}{n_\ell(x)}}
	\\ &=  \sqrt{\frac{1}{n_g(x)} +  \frac{1}{\max\{n_\ell(x): \ell \neq g\}}}. \nonumber
\end{align}
Moreover, let
$$
\mathcal S_{gN}(x) := \left \{i \in \mathcal G_{gN}(x): \|\bar U_{i \cdot} - (\mathcal UQ_U)_{i \cdot}\| \geq \frac 1 2  \sqrt{\frac{1}{n_g(x)} +  \frac{1}{\max\{n_\ell(x): \ell \neq g\}}}\right \}.
$$
Then 
\begin{align*}
	\|\bar U - \mathcal UQ_U\|_F^2 &= 	\sum_{g=1}^G \frac{1}{4}\left[\frac{1}{n_g(x)} +  \frac{1}{\max\{n_\ell(x): \ell \neq g\}}\right] \left(\frac{4   }{  \frac{1}{n_g(x)} +  \frac{1}{\max\{n_\ell(x): \ell \neq g\}}  }\sum_{i \in  \mathcal G_{gN}(x)} \|\bar U_{i \cdot} - (\mathcal UQ_U)_{i \cdot}\|^2\right)
	\\
	&\geq \sum_{g=1}^G \frac{1}{4}\left[\frac{1}{n_g(x)} +  \frac{1}{\max\{n_\ell(x): \ell \neq g\}}\right]\\
	& \qquad \qquad \left(\sum_{i \in  \mathcal G_{gN}(x)} {\bf 1}\left\{ \|\bar U_{i \cdot} - (\mathcal UQ_U)_{i \cdot}\| \geq \frac 1 2  \sqrt{\frac{1}{n_g(x)} +  \frac{1}{\max\{n_\ell(x): \ell \neq g\}}} \right \}\right)
	\\
	&= \sum_{g=1}^G |\mathcal S_{gN}(x)|\frac{1}{4}\left[\frac{1}{n_g(x)} +  \frac{1}{\max\{n_\ell(x): \ell \neq g\}}\right].
\end{align*}
Therefore we have
\begin{align*}
	\sum_{g=1}^G \frac{|\mathcal S_{gN}(x)|}{n_g(x)}& \leq \sum_{g=1}^G |\mathcal S_{gN}(x)| \left( \frac{1}{n_g(x)} +  \frac{1}{\max\{n_\ell(x): \ell \neq g\}} \right)
	\\
	& \leq 4 \|\bar U - \mathcal UQ_U\|_F^2
	\\
	& \leq 16 \|U - \mathcal UQ_U\|_F^2.
\end{align*}
By Lemma \ref{lem:davis} it follows that
\begin{equation}\label{eq:misclustRate}
	\sum_{g=1}^G \frac{|\mathcal S_{gN}(x)|}{n_g(x)} \leq \frac{512G}{\lambda_G(x,x')^2   }\| \widetilde{L_\tau^\eta} -  \widetilde{\mathcal L_\tau^\eta}   \|^2.
\end{equation}
The same holds with $  \frac{|\mathcal S_{hN}(x')|}{n_h(x')}$.   Moreover, one can show that following the proof of  Lemma 5.3 of  \citeasnoun{lei2015consistency},  the definition of $S_{gN}(x)$, in view of \eqref{eq:z-diff}, guarantees that for each $g \in [G]$, under the condition
\begin{equation}\label{eq:clusteringConstraint}
	\frac{16\sqrt{2G}}{\lambda_G(x,x')   }\| \widetilde{L_\tau^\eta} -  \widetilde{\mathcal L_\tau^\eta}   \| <1,
\end{equation}
the estimated membership matrix $\widehat{\Theta}^\eta(x)$   assigns the correct membership for every $i \in  \mathcal G_{gN}(x) \setminus  \mathcal S_{gN}(x)$.   That is, for $\mathcal{G}_N(x) = \cup_{g\in [G]} (\mathcal G_{gN}(x) \setminus  \mathcal S_{gN}(x))$, there exists a permutation matrix $J(x)$ such that 
$(\widehat{\Theta}\eta (x))_{\mathcal{G}_N(x) \cdot} J(x) = (\Theta^\eta(x))_{\mathcal{G}_N(x) \cdot} $ and the same holds for $x'$.

\medskip

By \eqref{eq:defDcal} and \eqref{eq:decompPopLagr}, the singular values of ${{\mathcal L}^\eta}(x,x',{\bf g})  $ are that of $\mathcal{N}(x) \overline{{\mathcal O}}^{-1/2} B(x,x') \overline{{\mathcal Q}}^{-1/2} \mathcal{N}(x')$. Denote $\sigma_1(M) \geq ... \geq \sigma_r(M)$ the singular values of a matrix $M$ of rank $r$. 
We use Theorem 3.3.16 of \citeasnoun{horn1990matrix}, see also \citeasnoun{wang1997some}, and write
\begin{align*}
	\lambda_G(x,x')& = \sigma_G\left(\mathcal{N}(x) \overline{{\mathcal O}}^{-1/2} B(x,x') \overline{{\mathcal Q}}^{-1/2} \mathcal{N}(x')\right) \\
	&\geq \frac{ \sigma_G\left(\mathcal{N}(x) \overline{{\mathcal O}}^{-1/2} B(x,x')\right)}{\sigma_1\left( \left[\overline{{\mathcal Q}}^{-1/2} \mathcal{N}(x')\right]^{-1}\right)} = \frac{ \sigma_G\left(\mathcal{N}(x) \overline{{\mathcal O}}^{-1/2} B(x,x')\right)}{\sigma_1\left( \mathcal{N}(x')^{-1}\overline{{\mathcal Q}}^{1/2} \right)}\\
	& \geq  \frac{\sigma_G\left(B(x,x')\right)}{\sigma_1\left( \mathcal{N}(x')^{-1}\overline{{\mathcal Q}}^{1/2} \right) \sigma_1\left(\left[ \mathcal{N}(x) \overline{{\mathcal O}}^{-1/2} \right]^{-1}\right)} =  \frac{\sigma_G\left(B(x,x')\right)}{\sigma_1\left( \mathcal{N}(x')^{-1}\overline{{\mathcal Q}}^{1/2} \right) \sigma_1\left( \mathcal{N}(x)^{-1} \overline{{\mathcal O}}^{1/2} \right)}
\end{align*}
where we recall that the matrices $\mathcal{N}(x)$, $\overline{{\mathcal O}}$ and $\overline{{\mathcal Q}}$ are diagonal. Their singular values are equal to their diagonal coefficients, and
\begin{align*}
	\left(  \mathcal{N}(x)^{-1} \overline{{\mathcal O}}^{1/2} \right)_{gg}^2 = \left(\sum_{j \in \eta_N(x')} B_{gg(j)}(x,x') + \tau\right) /n_g(x) \leq \left(\|B(x,x')\|_{\max} k + \tau\right)/\min_h \underline{n_h}.
\end{align*}
This implies that 
\begin{align}\label{eq:bdsingularvalue}
	\lambda_G(x,x') \geq \frac{\sigma_G\left(B(x,x')\right)\min_h \underline{n_h}}{\|B(x,x')\|_{\max} k + \tau} .
\end{align}
We plug in \eqref{eq:bdsingularvalue} and the bounds obtained on the relevant quantities in \eqref{eq:clusteringConstraint}. This condition becomes
\begin{align}
	4\sqrt{\frac{3 \ln (8k/\tilde{\delta})}{\Delta \min_{h \in [G]} \left \lfloor{\frac{c}{16} \frac{N_h}{N} \frac{b_X}{\overline{U}_X} \ k}\right \rfloor + \tau}}  + & \frac{2 k l_B R_k}{ \Delta \min_{h \in [G]} \left \lfloor{\frac{c}{16} \frac{N_h}{N} \frac{b_X}{\overline{U}_X} \ k}\right \rfloor + \tau} \left( \frac{2 k l_B R_k}{ \Delta \min_{h \in [G]} \left \lfloor{\frac{c}{16} \frac{N_h}{N} \frac{b_X}{\overline{U}_X} \ k}\right \rfloor+ \tau} + 3\right)\nonumber\\
	& < \ 16\sqrt{2G}\sigma_G\left(B(x,x')\right)\frac{\min_{h \in [G]} \left \lfloor{\frac{c}{16} \frac{N_h}{N} \frac{b_X}{\overline{U}_X} \ k}\right \rfloor}{\|B(x,x')\|_{\max} k + \tau}. \label{eq:clusteringConstraintDeterm}
\end{align}
We combine \eqref{eq:misclustRate} and \eqref{eq:clusteringConstraintDeterm} to obtain a finite sample probability bound on $\sum_{g=1}^G |\mathcal S_{gN}(x)|/n_g(x)$ in the following Lemma.

\begin{lem}\label{lem:clustRate}
	Let Assumptions \ref{assn:boundsupport}, \ref{assn:boundinff}, \ref{assn:boundsupf} and \ref{assn:LipshitzB} hold. Then for all $N$, $\delta \in (0,1)$, $1 \leq k \leq N$ such that	
	\begin{enumerate}
		\item $\Delta \min_{h \in [G]} \left \lfloor{\frac{c}{16} \frac{N_h}{N} \frac{b_X}{\overline{U}_X} \ k}\right \rfloor + \tau \geq 3 \ln (24k/\delta)$,
		\item $k \geq \max(12 d \ln(72GN/\delta),24 d \ln(36N/\delta))$,
		\item $k \leq \min(8  T^d  V_d \overline{U}_X N, (1/2)T^d  V_d b_X c N )$,
		\item and for all $h \in [G],$ $\frac{c}{16} \frac{N_h}{N} \frac{b_X}{\overline{U}_X} \ k \geq 24 d \ln(72 G N_h/\delta)  + 1$,
		\item \eqref{eq:clusteringConstraintDeterm} holds for $\tilde{\delta} = \delta/3$
	\end{enumerate}
	with probability at least $1-\delta$ conditional on ${\bold g}$, it holds that
	\begin{align}
		\sum_{g=1}^G \frac{|\mathcal S_{gN}(x)|}{n_g(x)} \leq & \frac{512G \left[\|B(x,x')\|_{\max} k + \tau\right]^2}{\sigma_G\left(B(x,x')\right)^2\min_{h \in [G]} \left \lfloor{\frac{c}{16} \frac{N_h}{N} \frac{b_X}{\overline{U}_X} \ k}\right \rfloor^2 } \left[
		 4\sqrt{\frac{3 \ln (24k/\delta)}{\Delta \min_{h \in [G]} \left \lfloor{\frac{c}{16} \frac{N_h}{N} \frac{b_X}{\overline{U}_X} \ k}\right \rfloor + \tau}} \right. \nonumber \\
		& \left. \qquad \qquad + \frac{2 k l_B R_k}{ \Delta \min_{h \in [G]} \left \lfloor{\frac{c}{16} \frac{N_h}{N} \frac{b_X}{\overline{U}_X} \ k}\right \rfloor + \tau} \left( \frac{2 k l_B R_k}{ \Delta \min_{h \in [G]} \left \lfloor{\frac{c}{16} \frac{N_h}{N} \frac{b_X}{\overline{U}_X} \ k}\right \rfloor+ \tau} + 3\right) \right]^2, \label{eq:Srate}
	\end{align}
	and that there exist permutation matrices $J(x)$ and $J(x')$ such that 
	$(\widehat{\Theta}\eta (x))_{\mathcal{G}_N(x) \cdot} J(x) = (\Theta^\eta(x))_{\mathcal{G}_N(x) \cdot} $ and the same holds for $x'$.
\end{lem}

\begin{proof}
	We use the events $V({\bf g},{\bf x})$ and  $W({\bf g},{\bf x})$ as well as the function $C$ defined in the proof of Lemma \ref{lem:integr}.  Let $\tilde{C}$ be the function such that $\tilde{C}(d_{\min})$ is the right hand side of \eqref{eq:Srate}.
	\begin{align*}
		\Pr &\left[  \sum_{g=1}^G \frac{|\mathcal S_{gN}(x)|}{n_g(x)}  \leq \tilde{C}\left(\Delta \min_{h \in [G]} \left \lfloor{\frac{c}{16} \frac{N_h}{N} \frac{b_X}{\overline{U}_X} \ k}\right \rfloor  \right) \,  \Big| \, {\bf g } \right] \\
		& \geq \Pr \left[  \left\lbrace \|\widetilde{L_\tau^\eta} - \widetilde{{\mathcal L}_\tau^\eta}(x,x',{\bf g}) \|  \leq C\left(\Delta \min_{h \in [G]} \left \lfloor{\frac{c}{16} \frac{N_h}{N} \frac{b_X}{\overline{U}_X} \ k}\right \rfloor  \right) \right\rbrace \cap V({\bf g},{\bf x}) \cap W({\bf g},{\bf x}) \,  \Big| \, {\bf g } \right]
	\end{align*}
	by \eqref{eq:misclustRate} and \eqref{eq:bdsingularvalue}. By the proof of Lemma \ref{lem:integr}, the probability on the right hand side is larger than $1-\delta$.
\end{proof}

\begin{rem}
	Let $\tau = 0$, sparsity be captured by a parameter $\rho_N$ and $N_h/N \approx C$ for all $h \in [G]$. Remark \ref{rem:opt_k} explains that $\|\widetilde{L^\eta} - \widetilde{{\mathcal L}^\eta}(x,x',{\bf g}) \| \approx (\rho_N N)^{\frac{-1}{d+2}}$, up to $\ln N$ and $\ln \rho_N$ factors. According to Lemma \ref{lem:clustRate}, this implies that 
	$$
	\sum_{g=1}^G \frac{|\mathcal S_{gN}(x)|}{n_g(x)} \approx (\rho_N N)^{\frac{-2}{d+2}},
	$$
	whereas \citeasnoun{lei2015consistency} obtains a rate of $(\rho_N N)^{-1}$, see Corollary 3.2.
\end{rem}

\
\section{Estimation of $B$ and $\pi$}

As in the previous section, we fix a pair $(x,x') \in S^2$. Recall that $\mathcal{G}_{h N}(x) = \{g(i) = h, i \in \eta_N(x)\}$ and $n_h(x) = |\mathcal{G}_{h N}(x)|$. We introduce the following new notation,
\begin{itemize}
	\item $\mathcal{G}_{h N}(x) = \{g(i) = h, i \in \eta_N(x)\}$, then $n_h(x) = |\mathcal{G}_{h N}(x)|$
	\item $\widehat{\mathcal{G}}_{h N}(x) = \{\hat{g}(i) = h, i \in \eta_N(x)\}$, then $\hat{n}_h(x)  = |\widehat{\mathcal{G}}_{h N}(x)|$
	\item $\mathcal{G}_{h N} = \{g(i) = h, i \in [N]\}$, then $N_h = |\mathcal{G}_{h N}|$
	\item $\bold{1}_{h,x}(i) = \bold{1}\{i \in \mathcal{G}_{h N}(x)\} $
	\item $\widehat{\bold{1}}_{h,x}(i) = \bold{1}\{i \in \widehat{\mathcal{G}}_{h N}(x)\} $
	\item $\bold{1}_{x}(i) = \bold{1}\{ i \in \eta_N(x)\} $
\end{itemize}
Note that the estimators of $\pi_h(x)$ and $B_{gh}(x,x')$ can be written
\begin{align*}
	\widehat{\pi}_h(x) & = \frac{\hat{n}_h(x)}{k} = \frac{\sum_{i\in \eta_N(x)} \widehat{\bold{1}}_{h,x}(i)}{k},\\
	\widehat{B}_{gh}(x,x') & = \frac{1}{\hat{n}_g(x) \hat{n}_h(x')} \ \sum_{\substack{i \in \widehat{\mathcal{G}}_{g N}(x)\\ j \in \widehat{\mathcal{G}}_{h N}(x')}}  A_{ij} =  \frac{1}{\hat{n}_g(x) \hat{n}_h(x')} \ \sum_{i,j}  A_{ij} \widehat{\bold{1}}_{g,x}(i) \widehat{\bold{1}}_{h,x'}(j).
\end{align*}
Define the oracle estimators for $\pi_g(x)$ and $B_{gh}(x,x')$ as 
\begin{align*}
	\pi^{or}_g(x) &= \frac{n_g(x)}{k} = \frac{\sum_{i} \bold{1}_{g,x}(i)}{k},\\
	B_{gh}^{or}(x,x') & = \frac{1}{n_g(x) n_h(x')} \ \sum_{i,j}  A_{ij} \bold{1}_{g,x}(i) \bold{1}_{h,x'}(j).
\end{align*}
In this section, we derive probability bounds on $\widehat{\pi}_h(x)$ and $\widehat{B}_{gh}(x,x')$. 
Note that according to Lemma \ref{lem:clustRate}, with probability  at least $1-\delta$ conditional on ${\bold g}$, the memberships of $i \in \left(\mathcal{G}_{g N}(x) \setminus \mathcal{S}_{gN}(x) \right)$ are all correctly estimated up to a permutation. Thus all the following probability bounds hold up to a permutation. We explore this identification issue this raises in Remark \ref{rem:ordering}.

\subsection{Result conditional on $\bold g$}

\

We can decompose
\begin{align*}
	|\widehat{B}_{gh}(x,x')  -  B_{gh}^{or}(x,x')|   \leq & \underbrace{\left| \frac{1}{\hat{n}_g(x) \hat{n}_h(x')} -  \frac{1}{n_g(x) n_h(x')}\right| \ \sum_{i,j}   \widehat{\bold{1}}_{g,x}(i) \widehat{\bold{1}}_{h,x'}(j)}_{T_1} \\
	&\qquad \qquad +  \underbrace{\frac{1}{n_g(x) n_h(x')} \sum_{i,j}\left|  \widehat{\bold{1}}_{g,x}(i) \widehat{\bold{1}}_{h,x'}(j) - \bold{1}_{g,x}(i) \bold{1}_{h,x'}(j)\right|}_{T_2},
\end{align*}
where by a slight abuse of notation the summations are over $i \in \eta_N(x)$ and $j \in \eta_N(x')$. We look at $T_1$ and $T_2$ separately.
\begin{align*}
	T_1 & = \left| 1 - \frac{ \hat{n}_g(x) \hat{n}_h(x')}{n_g(x) n_h(x')}\right| = \frac{\left|n_g(x) n_h(x') - \hat{n}_g(x) \hat{n}_h(x')\right|}{n_g(x) n_h(x')}, \\
	T_2 & \leq \underbrace{\frac{1}{n_g(x) n_h(x')}  \sum_{i,j}   \left|  \widehat{\bold{1}}_{g,x}(i) - \bold{1}_{g,x}(i)\right| \widehat{\bold{1}}_{h,x'}(j) }_{T_{21}}
	+ \underbrace{\frac{1}{n_g(x) n_h(x')}  \sum_{i,j}    \bold{1}_{g,x}(i)\left|\widehat{\bold{1}}_{h,x'}(j) - \bold{1}_{h,x'}(j)\right|}_{T_{22}}.
\end{align*}
When \eqref{eq:clusteringConstraint} holds, 
\begin{align*}
	T_{21} &= \frac{\hat{n}_h(x')}{n_h(x')} \frac{1}{\hat{n}_h(x')}\sum_{j} \widehat{\bold{1}}_{h,x'}(j) \frac{1}{n_g(x)}\sum_i  \left|  \widehat{\bold{1}}_{g,x}(i) - \bold{1}_{g,x}(i)\right|\\
	& \leq \frac{\hat{n}_h(x')}{n_h(x')} \frac{1}{n_g(x)} \sum_{l \in [G]} |\mathcal{S}_{lN}(x)|,
\end{align*}
where the second inequality holds by the argument used in \eqref{eq:rate_piHat_piOr_1stStep}.
Similarly, when \eqref{eq:clusteringConstraint} holds we obtain
$$
T_{22} \leq \frac{1}{n_h(x')} \sum_{l \in [G]} |\mathcal{S}_{lN}(x')|.
$$ 
Note that when \eqref{eq:clusteringConstraint} holds, 
\begin{align*}
	T_1 & \leq \frac{\left|n_g(x) - \hat{n}_g(x)\right| n_h(x') + \hat{n}_g(x) \left|\hat{n}_h(x') - n_h(x')\right|}{n_g(x) n_h(x')} \\
	& \leq \frac{1}{n_g(x)} \sum_{l \in [G]} |\mathcal{S}_{lN}(x)| + \frac{\hat{n}_g(x)}{n_g(x)} \frac{1}{n_h(x')} \sum_{l \in [G]} |\mathcal{S}_{lN}(x')| \\
	&\leq  \left[\frac{k}{\min_{h \in [G]} \underline{n_h} }  + \frac{k^2}{\left(\min_{h \in [G]} \underline{n_h}\right)^2} \right] \sum_{l \in [G]} \frac{|\mathcal S_{lN}(x)|}{n_l(x)},
\end{align*}
where we used the same argument as in \eqref{eq:rate_piHat_piOr_1stStep} in the second inequality.
We use similar arguments to bound $T_{21}$ and $T_{22}$ and obtain
\begin{equation}\label{eq:bdGraphonOr}
	|\widehat{B}_{gh}(x,x')  -  B_{gh}^{or}(x,x')|  \leq 2  \left[\frac{k}{\min_{h \in [G]} \underline{n_h} }  + \frac{k^2}{\left(\min_{h \in [G]} \underline{n_h}\right)^2} \right] \sum_{l \in [G]} \frac{|\mathcal S_{lN}(x)|}{n_l(x)}.
\end{equation}
We now bound $|B_{gh}^{or}(x,x') -B_{gh}(x,x')|$. The oracle estimator is not a standard $k$-nearest neighbor estimator with, say, $k=n_g(x)n_h(x')$,  because $n_g(x)$ and $n_h(x')$ are not chosen by the statistician but random. Moreover, the neighborhood for $x$ is chosen separately from that of $x'$.
It is also not a Nadaraya-Watson estimator with uniform kernel as $r_{n_g(x)}^g (x) $ and $r_{n_h(x')}^h (x') $ are random. Write
$$A_{ij} = B_{g(i)g(j)}(x(i),x(j)) + \zeta_{ij},$$
then we can decompose
\begin{align}
	|B_{gh}^{or}(x,x') -B_{gh}(x,x')| \leq & \underbrace{\frac{1}{n_g(x) n_h(x')} \ \sum_{i,j}  \big| \left[B_{gh}(x(i),x(j)) - B_{gh}(x,x')\right]\big| \bold{1}_{g,x}(i) \bold{1}_{h,x'}(j)}_{T_3} \nonumber \\
	& \qquad +  \underbrace{\frac{1}{n_g(x) n_h(x')} \ \big|\sum_{i,j}  \zeta_{ij} \bold{1}_{g,x}(i) \bold{1}_{h,x'}(j) \big|}_{T_4}. \label{eq:decomp_Boracle}
\end{align}
In this decomposition, $T_3$ is a bias term and $T_4$ is a variance term. As in the proof of Lemma \ref{lem:boundLaplacians}, the bias term can be bounded using Assumption \ref{assn:LipshitzB},
\begin{equation}\label{eq:T3}
	T_3 \leq l_B [r_k(x) + r_{k}(x')]
\end{equation}
by $r_{n_g(x)}^g (x) \leq r_k(x)$ and $r_{n_h(x')}^g (x') \leq r_k(x')$. Note that the relevant regressor dimension is $d$. As for the variance term $T_4$,
note that we can rewrite 
$$
T_4 = \frac{1}{n_g(x) n_h(x')} \ \big|\sum_{\substack{i \in \mathcal{G}_{g N}(x)\\j \in \mathcal{G}_{h N}(x')}}  \zeta_{ij}  \big|,
$$
where  conditional on $(\bold x,\bold g)$, $\{ \zeta_{ij}, \ i \in \mathcal{G}_{gN}(x),\  j \in \mathcal{G}_{hN}(x') \}$ are independent  mean-zero bounded random variables. We apply a Hoeffding inequality for bounded random variables, see Theorem 2.2.6 in   \citeasnoun{vershynin2018high}, and obtain for any $a_0 \geq 0$, 
\begin{align}\label{eq:hoeffding}
	\Pr \left( \big|\sum_{i,j}  \zeta_{ij} \bold{1}_{g,x}(i) \bold{1}_{h,x'}(j) \big| \geq a_0 \Bigg| \, \bold x,\bold g  \right) \leq 2 \exp\left( -a_0^2 / [2n_g(x) n_h(x')] \right),
\end{align}
where we used that  $\zeta_{ij} \in [-1,1]$. Taking $a_0 = \sqrt{2 n_g(x) n_h(x')\ln(2/\tilde{\delta})}$ implies
\begin{align}\label{eq:T4cond}
	\Pr \left( T_4 \leq \frac{\sqrt{2\ln(2/\tilde{\delta})}}{\min_{h \in [G]} \underline{n_h}} \, \Bigg| \, {\bf x}, \bold g  \right) \geq 1 -\tilde{\delta}.
\end{align}
We combine \eqref{eq:bdGraphonOr} with \eqref{eq:T3} and \eqref{eq:T4cond} to obtain the following Lemma.
\begin{lem}\label{lem:rate_BHat_g}
	Let Assumptions \ref{assn:boundsupport}, \ref{assn:boundinff}, \ref{assn:boundsupf} and \ref{assn:LipshitzB} hold. Then for all $N$, $\delta \in (0,1)$, $1 \leq k \leq N$ such that	
	\begin{enumerate}
		\item $\Delta \min_{h \in [G]} \left \lfloor{\frac{c}{16} \frac{N_h}{N} \frac{b_X}{\overline{U}_X} \ k}\right \rfloor + \tau \geq 3 \ln (48k/\delta)$,
		\item $k \geq \max(12 d \ln(144GN/\delta),24 d \ln(72N/\delta))$,
		\item $k \leq \min(8  T^d  V_d \overline{U}_X N, (1/2)T^d  V_d b_X c N )$,
		\item and for all $h \in [G],$ $\frac{c}{16} \frac{N_h}{N} \frac{b_X}{\overline{U}_X} \ k \geq 24 d \ln(144 G N_h/\delta)  + 1$,
		\item \label{cond:clusteringConstraint} \eqref{eq:clusteringConstraintDeterm} holds for $\tilde{\delta} = \delta/6$,
	\end{enumerate}
	with probability at least $1-\delta$ conditional on ${\bold g}$, it holds that
	\begin{align}
		&|\widehat{B}_{gh}(x,x') - B_{gh}(x,x')| \nonumber \\
		&\leq  \frac{1024G \left[\|B(x,x')\|_{\max} k + \tau\right]^2}{\sigma_G\left(B(x,x')\right)^2\min_{h \in [G]} \left \lfloor{\frac{c}{16} \frac{N_h}{N} \frac{b_X}{\overline{U}_X} \ k}\right \rfloor^2 }   \left[\frac{k}{\min_{h \in [G]} \left \lfloor{\frac{c}{16} \frac{N_h}{N} \frac{b_X}{\overline{U}_X} \ k}\right \rfloor  }  + \frac{k^2}{\left(\min_{h \in [G]} \left \lfloor{\frac{c}{16} \frac{N_h}{N} \frac{b_X}{\overline{U}_X} \ k}\right \rfloor \right)^2} \right] \times \nonumber \\
		&\left[
		4\sqrt{\frac{3 \ln (48k/\delta)}{\Delta \min_{h \in [G]} \left \lfloor{\frac{c}{16} \frac{N_h}{N} \frac{b_X}{\overline{U}_X} \ k}\right \rfloor + \tau}}  + \frac{2 k l_B R_k}{ \Delta \min_{h \in [G]} \left \lfloor{\frac{c}{16} \frac{N_h}{N} \frac{b_X}{\overline{U}_X} \ k}\right \rfloor + \tau} \left( \frac{2 k l_B R_k}{ \Delta \min_{h \in [G]} \left \lfloor{\frac{c}{16} \frac{N_h}{N} \frac{b_X}{\overline{U}_X} \ k}\right \rfloor+ \tau} + 3\right) \right]^2 \nonumber \\
		& \qquad  \qquad \qquad  \qquad   +  2 l_B R_k +  \frac{\sqrt{2\ln(4/\delta)}}{\min_{h \in [G]} \left \lfloor{\frac{c}{16} \frac{N_h}{N} \frac{b_X}{\overline{U}_X} \ k}\right \rfloor}. \label{eq:rate_BHat_g}
	\end{align}
\end{lem}
\begin{proof}
	We use once more the definitions introduced in the proof of Lemma \ref{lem:integr}.
	\begin{align*}
		\Pr \Bigg[ &  \left\lbrace \|\widetilde{L_\tau^\eta} - \widetilde{{\mathcal L}_\tau^\eta}(x,x',{\bf g}) \|  \leq C\left(\Delta \min_{h \in [G]} \left \lfloor{\frac{c}{16} \frac{N_h}{N} \frac{b_X}{\overline{U}_X} \ k}\right \rfloor  \right) \right\rbrace  \\
		& \qquad \qquad \left. \cap \left\lbrace T_4 \leq \frac{\sqrt{2\ln(4/\delta)}}{\min_{h \in [G]} \left \lfloor{\frac{c}{16} \frac{N_h}{N} \frac{b_X}{\overline{U}_X} \ k}\right \rfloor}  \right\rbrace \cap V({\bf g},{\bf x}) \cap W({\bf g},{\bf x}) \,  \Big| \, {\bf g } \right]\\
		&\geq \Pr \left[   \left\lbrace \|\widetilde{L_\tau^\eta} - \widetilde{{\mathcal L}_\tau^\eta}(x,x',{\bf g}) \|  \leq C\left(d_{\min} \right) \right\rbrace  \cap \left\lbrace T_4 \leq \frac{\sqrt{2\ln(4/\delta)}}{\min_{h \in [G]} \underline{n_h}}  \right\rbrace \cap V({\bf g},{\bf x}) \cap W({\bf g},{\bf x}) \,  \Big| \, {\bf g } \right]\\
		&\geq \Bigg[ \E\left[ \Pr \left[  \left\lbrace \|\widetilde{L_\tau^\eta} - \widetilde{{\mathcal L}_\tau^\eta}(x,x',{\bf g}) \|  \leq C\left(d_{\min} \right) \right\rbrace  \,  \Big| \, {\bf g }, {\bf x },  V({\bf g},{\bf x}), W({\bf g},{\bf x})\right]  \,  \Big| \, {\bf g },   V({\bf g},{\bf x}), W({\bf g},{\bf x}) \right] \\
		& \qquad \qquad \times \Pr \left[ V({\bf g},{\bf x}) \cap W({\bf g},{\bf x})  \,  | \, {\bf g }\right]  + \E\left[ \Pr \left[ \left\lbrace T_4 \leq \frac{\sqrt{2\ln(4/\delta)}}{\min_{h \in [G]} \underline{n_h}}  \right\rbrace \,  \Big| \, {\bf g }, {\bf x }\right]  \,  \Big| \, {\bf g }   \right]   \ - \ 1 \\
		& \geq 1 - \delta/2 + 1 - \delta/2 -1 = 1 - \delta
	\end{align*}
	where the last inequality holds by \eqref{eq:T4cond} and the proof of Lemma \ref{lem:integr}.
	The result holds by
	\begin{align*}
		&\left\lbrace \|\widetilde{L_\tau^\eta} - \widetilde{{\mathcal L}_\tau^\eta}(x,x',{\bf g}) \|  \leq C\left(\Delta \min_{h \in [G]} \left \lfloor{\frac{c}{16} \frac{N_h}{N} \frac{b_X}{\overline{U}_X} \ k}\right \rfloor  \right) \right\rbrace   \cap V({\bf g},{\bf x}) \cap W({\bf g},{\bf x}) \\
		&\Rightarrow \begin{cases}
			& \sum_{l \in [G]} \frac{|\mathcal S_{lN}(x)|}{n_l(x)} \leq \epsilon_N(\delta/6),\\
			&	|\widehat{B}_{gh}(x,x')  -  B_{gh}^{or}(x,x')|  \leq 2  \left[\frac{k}{\min_{h \in [G]} \left \lfloor{\frac{c}{16} \frac{N_h}{N} \frac{b_X}{\overline{U}_X} \ k}\right \rfloor  }  + \frac{k^2}{\left(\min_{h \in [G]} \left \lfloor{\frac{c}{16} \frac{N_h}{N} \frac{b_X}{\overline{U}_X} \ k}\right \rfloor \right)^2} \right] \epsilon_N(\delta/6),\\
			&T_3 \leq 2 l_B R_k,
		\end{cases}
	\end{align*}
	where the first line holds as in Lemma \ref{lem:clustRate} under condition (\ref{cond:clusteringConstraint}), the second holds by \eqref{eq:bdGraphonOr}, and the third by \eqref{eq:T3}.
\end{proof}

\subsection{Results unconditional on ${\bf g}$}

\subsubsection{Result on the community assignment probabilities}

\

The following results hold up to a permutation but for ease of readability we let $J(x) = J(x')=I$. 
Note that when \eqref{eq:clusteringConstraint} holds,
\begin{align}
	|\widehat{\pi}_g(x) - \pi^{or}_g(x)| & \leq \frac{1}{k}  \sum_i \left| \widehat{\bold{1}}_{g,x}(i) - \bold{1}_{g,x}(i)  \right|  \nonumber \\
	& \leq \frac{1}{k} \sum_{\substack{h \in [G]\\ h \neq g}} \sum_{i \in \mathcal{G}_{h N}(x)}  \widehat{\bold{1}}_{g,x}(i)   +  \frac{1}{k} \sum_{i \in \mathcal{G}_{g N}(x)} \left| \widehat{\bold{1}}_{g,x}(i)  - 1  \right| \label{eq:rate_piHat_piOr_1stStep}\\
	& \label{eq:rate_piHat_piOr} \leq \frac{1}{k} \sum_{h \in [G]} |\mathcal{S}_{hN}(x)| \leq \sum_{h \in [G]} \frac{|\mathcal S_{hN}(x)|}{n_h(x)} 
\end{align}

The behavior of $\pi^{or}_g(x)$ is given by Theorem 1 of \citeasnoun{jiang2019non}. Write 
$$
\bold{1}_{g,x}(i) =\pi_g(x) +\xi_i
$$
with $\E(\xi_i|x(i)=x)=0$. Assumptions 1-3 of \citeasnoun{jiang2019non} hold by Assumptions \ref{assn:boundsupport}, \ref{assn:boundinff} and taking the sub-gaussian parameter to be $ 1$ since $\xi_i \in [-1,1]$ almost surely, see Exercise 2.4 in \citeasnoun{wainwright2019high}. We note that the imposed assumption of independence between $\xi$ and $x$ is not needed for his Theorem 1\footnote{The application of Hoeffding's inequality in the Proof of Theorem 1 (see p4004) can be done conditional on $\bold x$ as long as $(x(i),\xi_i)$ for $i=1...n$ is i.i.d.}.
Assume also that 
\begin{equation}\label{eq:jiangCstrnt}
	2^8 d \ln(4/\tilde{\delta}) \ln N \leq k \leq c V_d b_X T^d N/2,
\end{equation}
then under Assumption \ref{assn:LipshitzPi}, Theorem 1 of \citeasnoun{jiang2019non} implies that the following holds with probability at least $1 - \tilde{\delta}$,
\begin{equation}\label{eq:rate_piOracle}
	|\pi^{or}_g(x) - \pi_g(x)| \leq l_\pi R_k + 2  \sqrt{\frac{d \ln N + \ln(2/\tilde{\delta})}{k}}.
\end{equation}
Let $\pi_g := \Pr(g(i)=g) = \E(\pi_g(x(i)))$ and $\underline{\pi} := \min_{g \in [G]} \pi_g$. We combine Equations \eqref{eq:rate_piHat_piOr} and \eqref{eq:rate_piOracle} to obtain a probability bound unconditional on ${\bf g}$ in the following Lemma, where \eqref{eq:clusteringConstraintDeterm} is replaced with
\begin{align}
	4\sqrt{\frac{3 \ln (24k/\tilde{\delta})}{\Delta \left \lfloor\frac{\underline{\pi} c b_X}{32\overline{U}_X} \ k\right \rfloor + \tau}}  + & \frac{2 k l_B R_k}{ \Delta \left \lfloor\frac{\underline{\pi} c b_X}{32\overline{U}_X} \ k\right \rfloor + \tau} \left( \frac{2 k l_B R_k}{\Delta \left \lfloor\frac{\underline{\pi} c b_X}{32\overline{U}_X} \ k\right \rfloor + \tau} + 3\right) \nonumber\\
	&\qquad  \qquad <\ \frac{16\sqrt{2G}\sigma_G\left(B(x,x')\right)\left \lfloor\frac{\underline{\pi} c b_X}{32\overline{U}_X} \ k\right \rfloor }{\|B(x,x')\|_{\max} k + \tau}. \label{eq:clusteringConstraintIntegr}
\end{align}
\begin{lem}\label{lem:rate_piHat}
	Let Assumptions \ref{assn:boundsupport}, \ref{assn:boundinff}, \ref{assn:boundsupf}, \ref{assn:LipshitzB} and  \ref{assn:LipshitzPi} hold. Assume moreover that $\underline{\pi} > 0$. Then for all $N$, $\delta \in (0,1)$, $1 \leq k \leq N$ such that		
	\begin{enumerate}
		\item  $\Delta \left \lfloor\frac{\underline{\pi} c b_X}{32\overline{U}_X} \ k\right \rfloor + \tau \geq 3 \ln (72k/\delta)$,
		\item  $k \geq \max(12 d \ln(216GN/\delta),24 d \ln(108N/\delta))$,
		\item $k \leq \min(8  T^d  V_d \overline{U}_X N, (1/2)T^d  V_d b_X c N )$,
		\item $\frac{\underline{\pi} c b_X}{32\overline{U}_X} \ k \geq 24 d \ln(216 G N/\delta)  + 1$,
		\item \eqref{eq:clusteringConstraintIntegr} holds for $\tilde{\delta} = \delta/3$,
		\item $2^8 d \ln(24/\delta) \ln N \leq k \leq c V_d b_X T^d N/2,$	
		\item $N \geq 8 \ln(3G/\delta) / \underline{\pi}^2$,
	\end{enumerate}
	with probability at least $1-\delta$, it holds that
	\begin{align}
		|\widehat{\pi}_g(x) - \pi_g(x)| \leq &\frac{512G \left[\|B(x,x')\|_{\max} k + \tau\right]^2}{\sigma_G\left(B(x,x')\right)^2\left \lfloor\frac{\underline{\pi} c b_X}{32\overline{U}_X} \ k\right \rfloor^2 }  \left[
		4\sqrt{\frac{3 \ln (72k/\delta)}{\Delta \left \lfloor\frac{\underline{\pi} c b_X}{32\overline{U}_X} \ k\right \rfloor + \tau}}  + \frac{2 k l_B R_k}{ \Delta \left \lfloor\frac{\underline{\pi} c b_X}{32\overline{U}_X} \ k\right \rfloor + \tau} \right. \nonumber\\
		& \qquad \qquad \qquad \qquad \times \left. \left( \frac{2 k l_B R_k}{ \Delta \left \lfloor\frac{\underline{\pi} c b_X}{32\overline{U}_X} \ k\right \rfloor + \tau} + 3\right) \right]^2  +  l_\pi R_k + 2  \sqrt{\frac{d \ln N + \ln(6/\delta)}{k}}. \label{eq:rate_piHat}
	\end{align}
\end{lem}

\begin{proof}
	For ease of readability, we denote with $\epsilon_N(\delta)$ the bound on the right hand side of \eqref{eq:Srate}. We apply a Hoeffding inequality for bounded random variables, see Theorem 2.2.6 in   \citeasnoun{vershynin2018high}, and obtain, 
	\begin{align*}
		\Pr \left(N_g - \pi_g N \geq -\pi_g N / 2 \right)  & = \Pr \left( \sum_{i \in [N] } \bold{1}\{ g(i) = g\} - \pi_g \geq - \pi_g N / 2 \right) \\
		& \leq  \exp \left( -  \pi_g^2 N / 8\right)  \leq  \exp \left( - \underline{\pi}^2 N / 8\right).
	\end{align*}
	Define the event $T({\bf g}) = \{\forall \, g \in [G], \ N \geq N_g \geq \underline{\pi} N/2\}$. Then 
	$$
	\Pr\left(T({\bf g})\right) \geq 1 - G \exp \left( -  \underline{\pi}^2 N / 8\right) \geq 1 - \delta/3,
	$$
	by Condition (7). Note that on $T({\bf g})$, Conditions (1) - (5) of Lemma \ref{lem:clustRate} hold. Thus,
	$$
	\Pr\left( \sum_{h \in [G]} \frac{|\mathcal S_{hN}(x)|}{n_h(x)} \leq \epsilon_N(\delta/3) \bigg| {\bf g}, T(\bf g)\right) \geq 1 - \delta/3.
	$$
	Moreover on $T({\bf g})$, note that the upper bound on $\epsilon_N(\delta/3)$ is bounded above by
	\begin{align*}
		\frac{512G \left[\|B(x,x')\|_{\max} k + \tau\right]^2}{\sigma_G\left(B(x,x')\right)^2\left \lfloor\frac{\underline{\pi} c b_X}{32\overline{U}_X} \ k\right \rfloor^2 } \left[
		4\sqrt{\frac{3 \ln (72k/\delta)}{\Delta \left \lfloor\frac{\underline{\pi} c b_X}{32\overline{U}_X} \ k\right \rfloor + \tau}}  + \frac{2 k l_B R_k}{ \Delta \left \lfloor\frac{\underline{\pi} c b_X}{32\overline{U}_X} \ k\right \rfloor + \tau} \left( \frac{2 k l_B R_k}{ \Delta \left \lfloor\frac{\underline{\pi} c b_X}{32\overline{U}_X} \ k\right \rfloor + \tau} + 3\right) \right]^2.
	\end{align*}
	The result follows by  Theorem 1 of \citeasnoun{jiang2019non}.	
\end{proof}

\

\subsubsection{Result on the connection probabilities}

\

We integrate Lemma \ref{lem:rate_BHat_g} with respect to ${\bf g}$ on $T({\bf g})$ and obtain the following Lemma.
\begin{lem}\label{lem:rate_BHat}
	Let Assumptions \ref{assn:boundsupport}, \ref{assn:boundinff}, \ref{assn:boundsupf} and \ref{assn:LipshitzB} hold. Assume moreover that $\underline{\pi} > 0$. Then for all $N$, $\delta \in (0,1)$, $1 \leq k \leq N$ such that	
	\begin{enumerate}
		\item $\Delta \left \lfloor\frac{\underline{\pi} c b_X}{32\overline{U}_X} \ k\right \rfloor + \tau \geq 3 \ln (96k/\delta)$,
		\item $k \geq \max(12 d \ln(288GN/\delta),24 d \ln(144N/\delta))$,
		\item $k \leq \min(8  T^d  V_d \overline{U}_X N, (1/2)T^d  V_d b_X c N )$,
		\item  $\frac{\underline{\pi} c b_X}{32\overline{U}_X} \ k \geq 24 d \ln(288 G N/\delta)  + 1$,
		\item \eqref{eq:clusteringConstraintIntegr} holds for $\tilde{\delta} = \delta/12$,
		\item $N \geq 8 \ln(2G/\delta) / \underline{\pi}^2$,
	\end{enumerate}
	with probability at least $1-\delta$ conditional on ${\bold g}$, it holds that
	\begin{align}
		&|\widehat{B}_{gh}(x,x') - B_{gh}(x,x')| \nonumber \\
		&  \leq  \frac{1024G \left[\|B(x,x')\|_{\max} k + \tau\right]^2}{\sigma_G\left(B(x,x')\right)^2\min_{h \in [G]}\left \lfloor\frac{\underline{\pi} c b_X}{32\overline{U}_X} \ k\right \rfloor^2 }   \left[\frac{k}{\left \lfloor\frac{\underline{\pi} c b_X}{32\overline{U}_X} \ k\right \rfloor}  + \frac{k^2}{\left \lfloor\frac{\underline{\pi} c b_X}{32\overline{U}_X} \ k\right \rfloor^2} \right] \times \label{eq:rate_BHat}\\
		&\qquad \qquad \left[
		4\sqrt{\frac{3 \ln (48k/\delta)}{\Delta\left \lfloor\frac{\underline{\pi} c b_X}{32\overline{U}_X} \ k\right \rfloor + \tau}}  + \frac{2 k l_B R_k}{ \Delta \left \lfloor\frac{\underline{\pi} c b_X}{32\overline{U}_X} \ k\right \rfloor + \tau} \left( \frac{2 k l_B R_k}{ \Delta \left \lfloor\frac{\underline{\pi} c b_X}{32\overline{U}_X} \ k\right \rfloor + \tau} + 3\right) \right]^2 \\
		& \ \ + 2 l_B R_k +  \frac{\sqrt{2\ln(4/\delta)}}{\left \lfloor\frac{\underline{\pi} c b_X}{32\overline{U}_X} \ k \right \rfloor}. \nonumber
	\end{align}
\end{lem}

\begin{proof}
	Under Condition (6), $\Pr(T({\bf g})) \geq 1 - \delta/2$ and under the remaining conditions, by Lemma  \ref{lem:rate_BHat_g},
	$$
	\Pr\left( \text{ \eqref{eq:rate_BHat} holds } | \, {\bf g}, T({\bf g}) \right) \geq 1 - \delta/2.
	$$	
\end{proof}

\begin{rem}\label{rem:ordering}  As in the existing literature on the SBM (or other models with mixture structure in general) the preceding results hold up to relabeling of communities.   This usually causes no issues since community labels  have little practical implications in those models.    While this observation partially applies to our procedure, the issue of community labeling still  poses a novel challenge in terms of the interpretability of our estimators. 
	
	Take a simple, special case with $G = 2$.  The results obtained above enable us to identify and estimate each of  the four elements of the edge probability matrix $B(x,x')$, $(x,x') \in S^2$, as well as the two vectors of community assignment probabilities, up to relabeling of the rows and the columns.  We are clearly free to choose the labels for either the rows or the columns, so let us say we fix the two labels for the rows (those corresponding  to the nodes with covariate value being $x$): this essentially amounts to normalization.    Given this normalization, however, one may wish to identify the community labels for the columns (i.e. the order of the two columns of $B$), at least for two reasons.  First, even when one is interested in the edge probabilities and  the community assignment probabilities at just one point $(x,x')$ in $S^2$,  interpreting and using these  probabilities might demand identification of the column (row) order, relative to a given  choice of the row (column) order.  For example, a diagonal element  of the edge probability matrix represents connections within an unobservable community (though they generally differ in terms of the observed heterogeneity, as far as $x \neq x'$); Such within-community connections are often associated with homophily/heterophily  or (dis)assortativity.  Of course, such issues potentially affect off-diagonal elements of the edge probability matrix with general $G \geq 2$.  Second, if, for example, one is interested in the partial effect of moving $x'$ to $x''$ in $B_{gh}(x,\cdot)$ or $\pi_h(\cdot)$,  then it is obvious that we need to have the correspondence between the community labels  remain consistent between  $(x,x')$ and $(x,x'')$.       
	
	This issue does not arise in the standard SBM without covariates as in \citeasnoun{lei2015consistency}, as their edge probability matrix are defined for the same population.  This holds true even in the analysis of asymmetric networks in \citeasnoun{rohe2016co}.      In our analysis, if $x \neq x'$, then the edge probability matrix $B(x,x')$ is concerned with edges between two separate populations, and this feature gives rise to difficulties in guaranteeing  proper matching between the row-clusters and the column clusters  in the absence of further information/restrictions.    Of course, by setting $x = x'$ in $B(x,x')$,  matching the community labels between the two sides is trivial, and then to the extent that continuity of $B_{g,h}(\cdot,\cdot), (g,h) \in [G]\times[G]$ and the connectedness of the support of covariate $X$ permit, the labels can be matched as we move $x'$ away from $x$.  Such approach may fail to be reasonable or practical in actual applications, however.      
	
	If we are willing to impose additional restrictions, it is possible to address this issue directly.  For example, once again, we are free to choose community labeling  such that   $\pi_1(x) > \pi_2(x) > \cdots > \pi_G(x)$, assuming no ties.   Suppose we also choose community labels for the columns, so that  $\pi_1(x') > \pi_2(x') > \cdots > \pi_G(x')$.  
	Under the assumption that the ranking of the magnitudes of  community assignment probabilities is invariant between those at $x$ and those at $x'$,  then trivially the community labels on both sides can be matched in a consistent manner.      
	
	One may also wish to introduce qualitative restrictions that are motivated by concepts developed in the literature of network analysis,  in order to achieve successful community label matching.   Examples of such restrictions include homophily/heterophily, or, assortativity/disasortativity, {\emph {conditional on the covariates}}.    Suppose the edge probability matrix $B(x,x')$ is conditionally weakly assortative at $(x,x') \in S^2$, in the sense that\footnote{Technically, even a weaker condition such as 
		\begin{equation*}
			B_{gg}(x,x') >  B_{g,h}(x,x') {\text{ for every }} (g,h) \in [G] \times [G]  \text{ with } g \neq h.
		\end{equation*} 
		which ensures that each diagonal element dominates the rest of its row elements --- or its column elements, by symmetry ---  suffices.  The strong assortative version of  \eqref{eq:homophily}  can be obtained by strengthening  the inequality by  
		\begin{equation*}
			B_{gg}(x,x') >  B_{h,f}(x,x')      {\text{ for every }} (f,g,h) \in [G] \times [G] \times [G] \text{ with } h \neq f.
		\end{equation*}         
	} 
	\begin{equation}\label{eq:homophily}
		B_{gg}(x,x') >  \max ( B_{g,h}(x,x'),B_{h,g}(x,x'))      {\text{ for every }} (g,h) \in [G] \times [G]  \text{ with } h \neq f.
	\end{equation}            
	This can be justified under homophily in terms of unobserved community membership, while the impact of covariates $(x,x')$ on edge probabilities remain fully unspecified (and can be correlated with unobserved heterogeneity in an arbitrary manner).  The restriction \eqref{eq:homophily} suffices to achieve identification of matched community labels.      
	Let $\mathcal E_G$ denote the set of $G \times G$ permutation matrices.  Without loss of generality, assign $G$ labels on the rows of the edge probability matrix; let $B(x,x')$ be the resulting matrix to be (uniquely) recovered.  Without a restriction such as \eqref{eq:homophily}, we can only identify the set $\{B(x,x')E_G, E_G \in \mathcal{E}_G\}$.  Let $B^\circ(x,x')$ an arbitrary element of the set.   To exploit the asortativity restriction \eqref{eq:homophily}, we can simply solve       
	$$
\overline	E_G = \argmax_{ E_G \in \mathcal{E}_G}  \tr[B^\circ(x,x')E_G], 
	$$
	and then we recover $B(x,x')$  by $B^\circ(x,x')\overline E_G$.   This offers a practical algorithm, as the maximization problem  $\max_{ E_G \in \mathcal{E}_G}  \tr[A E_G]$ is well defined for any $G$ by $G$ matrix $A$ as far as the maximum entry of each row has no ties.   Likewise, if one is willing to impose a heterophily, we can flip the inequality sign in  \eqref{eq:homophily}  to restrict $B(x,x')$ to be disasortative, then     solve
	$$
	\underline{E}_G = \argmin_{ E_G \in \mathcal{E}_G}  \tr[B^\circ(x,x')E_G], 
	$$ 
to	recover the desired edge probability matrix.

\end{rem}

\section{Conclusion}  This paper demonstrates that it is possible to incorporate both observed and unobserved heterogeneity in network data  analysis in a flexible way, at least when we have discrete values of heterogeneity represented  by community assignments.  It offers a highly versatile, yet computationally tractable procedure, which is expected to complement the existing methodology for analyzing networks with covariates under  more specific structures for the edge probabilities and the community assignment probabilities.    Our results build upon  recent developments in spectral clustering in SBMs and $k$-nn algorithms, and we contribute to the literature by addressing novel theoretical challenges presented by our multi-step procedure.  Though  our estimators can be computed in a straightforward manner, an extensive simulation exercise is called for in order to assess the efficacy of the new procedure in a practical setting.

\
\

\renewcommand{\thesection}{A \arabic{section}}

\setcounter{section}{0}
\section*{Supplement: Some  Useful results}
\setcounter{section}{1}

\subsection{Nearest Neighbor Radius}

\

To obtain uniform upper and lower bounds on the radiuses, we use inequalities for relative deviations see \citeasnoun{anthony1993result} and Section 1.4.2 of \citeasnoun{lugosi2002pattern}. 
Before stating the two inequalities we will use, we introduce some definitions. 
Let $x$ denote a ${\bf R}^d$ valued covariate vector. 
Let  $P_{{\bf x} N}$ denote the empirical measure based on ${\bf x} = (x_1,...,x_N)$, that is, $P_{{\bf x} N}(A) := \#\{x(i) \in A, i \in [N] \}/N$  for $A \in \mathcal C$.  Likewise define $P_{{\bf y} N}(A) := \#\{y(i) \in A, i \in [N] \}/N$ and $P_{{\bf xy} N}(A) := \#\{x(i) \in A, y(i) \in A, i \in [N] \}/2N$.  We consider  ${\bf x} = (x_1,...,x_N)$ and ${\bf y} = (y_1,...,y_N)$ defined on the sample space $S^N$, independent of each other and both distributed according to  $P(.)^n$ for a probability measure $P$.   Let $\mathcal C$ be a collection of subsets of $S$.   For a sample ${\bf x} = (x_1,...,x_N)$ viewed as a collection of draws $\{x_1,...,x_N\}$, we define as in \citeasnoun{gine2021mathematical} (Section 3.6.1) the trace of $\mathcal{C}$ on $\bf x$ as all the subsamples of $\bf x$ obtained by intersection of  $\bf x$  with sets $A \in \mathcal{C}$. Define $\Delta_{\mathcal{C}}( {\bf x})$ as the cardinal of the trace of the collection $\mathcal{C}$ and
$$
m_{\mathcal{C}}(N) = \sup_{{\bf x} \in S^N } \Delta_{\mathcal{C}}( {\bf x}).
$$
$m_{\mathcal{C}}(N)$ is the shattering coefficient of the collection $\mathcal{C}$.
Following \citeasnoun{anthony1993result}, we define a complete set of distinct representatives (CSDR) of $\mathcal{C}$ for $\bf x$ as a collection $\mathcal{A} = \{A^1,...,A^{\Delta_{\mathcal{C}}( {\bf x})}\}$ if for any $1 \leq i \neq j\leq \Delta_{\mathcal{C}}( {\bf x})$ then $A^i \cap \{x_1,...,x_N\} \neq A^j \cap \{x_1,...,x_N\}$. For all $A \in \mathcal{C}$, there exists $1 \leq i \leq \Delta_{\mathcal{C}}( {\bf x})$ such that $A \cap \{x_1,...,x_N\} = A^i \cap \{x_1,...,x_N\}$.

We are in particular interested in the following inequalities, see e.g. Theorem 1.11 in \citeasnoun{lugosi2002pattern},
\begin{align}
	\forall \eta >0,\	& \Pr \left( \sup_{a \in \mathcal{C}} \frac{P(A) - P_{{\bf x} N}(A)}{\sqrt{P(A)}} > \eta \right) \leq 4 m_{\mathcal{C}}(2N)  \exp \left(-\eta^2 N / 4 \right) \label{eq:relative1}\\
	\forall \eta >0,\	& \Pr \left( \sup_{a \in \mathcal{C}} \frac{ P_{{\bf x} N}(A) - P(A)}{\sqrt{ P_{{\bf x} N}(A)}} > \eta \right) \leq 4 m_{\mathcal{C}}(2N)  \exp \left(-\eta^2 N / 4 \right) \label{eq:relative2}
\end{align}
where $\mathcal{C}$ is any collection of Borelian sets. We first give a proof of \eqref{eq:relative2}. As we could not find one in the literature, we restate one which we will later modify to accommodate non-identically (but independently) distributed random variables. 

\begin{lem}\label{le:proofrelative2}
	For $\mathcal{C}$ any collection of Borelian sets, \eqref{eq:relative2} holds.
\end{lem}

\begin{proof}
	The proof adapts the steps of \citeasnoun{anthony1993result}. We define the sets
	$$
	Q := \left \{(x_1,...,x_N) \in S^{N}: \exists A \in \mathcal{C} \text{ such that }  \frac{ P_{{\bf x} N}(A) - P(A)}{\sqrt{ P_{{\bf x} N}(A)}} > \eta \right    \} 
	$$
	$$
	R := \left \{(x_1,...,x_N,y_1,...,y_N) \in S^{2N}: \exists A \in \mathcal{C} \text{ such that } {P_{{\bf x} N}(A) - P_{{\bf y} N}(A)} > \eta{\sqrt{P_{{\bf xy} N}(A)}} \right    \}.   
	$$
	We first look at the case $N>2/\eta$.  First note that for each  ${\bf x} \in Q$ there exists a set $A_{\bf x} \in \mathcal{C}$, indexed by $\bf x$,  such that $ P_{{\bf x} N}(A_{\bf x}) - P(A_{\bf x}) > \eta \sqrt{P_{{\bf x} N}(A_{\bf x})}.$  Define 
	$$
	F_{\bf xy}(A_{\bf x}) :=  \frac{P_{{\bf x} N}(A_{\bf x}) - P_{{\bf y} N}(A_{\bf x})}{\sqrt{P_{{\bf xy} N}(A_{\bf x})}}.  
	$$
	Take $P_{{\bf y} N}(A_{\bf x})$ such that $P_{{\bf y} N}(A_{\bf x}) < P(A_{\bf x})$. Then 
	\begin{align*}
		F_{\bf xy}(A_{\bf x}) &> \frac{P(A_{\bf x}) + \eta \sqrt{P_{{\bf x} N}(A_{\bf x})} - P_{{\bf y} N}(A_{\bf x})}{\sqrt{[P_{{\bf x} N}(A_{\bf x})) + P_{{\bf y} N}(A_{\bf x})]/2}} \\
		&> \frac{\eta \sqrt{P_{{\bf x} N}(A)} }{\sqrt{[P_{{\bf x} N}(A_{\bf x}) + P(A_{\bf x})]/2}} > \eta,
	\end{align*}
	where the second inequality holds by monotonicity in $P_{{\bf x} N}(A_{\bf x})$ and the third  by $P_{{\bf x} N}(A_{\bf x}) > P(A_{\bf x})$.
	\begin{eqnarray*}
		\Pr\{({\bf x},{\bf y}) \in R |{\bf g}\}   &\geq&    \Pr\{ F_{\bf xy}(A_{\bf x}) > \eta\ |  P_{{\bf y} N}(A_{\bf x}) < P(A_{\bf x}), {\bf x} \in Q  \} \Pr\{P_{{\bf y} N}(A_{\bf x})  <   P(A_{\bf x}), {\bf x} \in Q\}
		\\
		&\geq&  \left(   \inf_{{\bf x} \in Q}     \Pr\{P_{{\bf y} N}(A_{\bf x}) < P(A_{\bf x})  \}   \right)   \Pr\{  {\bf x} \in Q \} = \left(   \inf_{{\bf x} \in Q}     \Pr\{P_{{\bf y} N}(A_{\bf x}^c) > P(A_{\bf x}^c)  \}   \right)   \Pr\{  {\bf x} \in Q \}
		\\
		&\geq&   \frac{1}{4}\ \Pr\{  {\bf x} \in Q \}.
	\end{eqnarray*}
	where $A_{\bf x}^c $ is the complement set of $A_{\bf x}$. The last inequality uses  Theorem 1 of   \citeasnoun{greenberg2014tight} thus we need to check that for all ${\bf x} \in Q $, $P(A_{\bf x}^c) > 1/N$, that is, $P(A_{\bf x}) < 1-1/N$.
	Note that $P(A_{\bf x}) < P_{{\bf x} N}(A_{\bf x}) - \eta \sqrt{P_{{\bf x} N}(A_{\bf x})} =: g(P_{{\bf x} N}(A_{\bf x}))$. The function $g$ decreases on $[0,\eta^2/4]$ and is negative on this interval, and increases on $[\eta^2/4, +\infty)$. Thus  $P(A_{\bf x}) < g(1) = 1-\eta$, which imposes $\eta <1$, and  $P(A_{\bf x}) < 1-1/N$ follows by $N>2/\eta$. The inequality \eqref{eq:relative2} is obtained by 
	$$\Pr\{({\bf x},{\bf y}) \in R |{\bf g}\}   \leq m_{\mathcal{C}}(2N)  \exp \left(-\eta^2 N / 4 \right),$$
	which holds by the proof of Theorem  2.1 in \citeasnoun{anthony1993result}.
		
	If $N<2/\eta$, the upper bound in \eqref{eq:relative2} is $4 m_{\mathcal{C}}(2N)  \exp \left(-\eta^2 N / 4 \right) \geq 4 m_{\mathcal{C}}(2N)  \exp \left(-\eta /2 \right)$. For ${\bold x} \in Q$, $P(A_{\bf x}) <  g(P_{{\bf x} N}(A_{\bf x}))$ guarantees that  $\eta <1$ which implies that $4 m_{\mathcal{C}}(2N)  \exp \left(-\eta^2 N / 4 \right) \geq 4  \exp \left(-1 /2 \right) \geq 2 $. Thus  \eqref{eq:relative2} holds naturally.  Note that if $Q = \emptyset$, \eqref{eq:relative2} holds naturally as well.
\end{proof}

Before elaborating on our results, we state two preliminary lemmas.

\begin{lem}\label{lem:1/4} Suppose $\{z_i\}_{i=1}^N$ are independently distributed, with $z_i \sim$Bernoulli$(q_i)$, $0 < q_i < 1$ for every $i \in [N]$.  Let $S := \sum_{i=1}^N z_i$, then $\Pr\{S > N \min_{i \in [N]}q_i\} > \frac{1}{4}$ if $ \min_{i \in [N]}q_i > \frac{1}{N}$.  
\end{lem}

\begin{proof}
	It is easy to see that, if $\underline{z}_i \sim_{iid} $Bernoulli$(\min_{i \in [N]}q_i)$, and we let $\underline{S} := \sum_{i=1}^{N} \underline{z}_i$, then $\Pr\{S > N \min_{i \in [N]}q_i\} \geq \Pr\{\underline{S} > N \min_{i \in [N]}q_i\}$.  Indeed, note that by definition $z_i$ first-order stochastically dominates $\underline{z}_i$, or $z_i \geq_1 \underline{z}_i $ for each $i$.  Then by Theorem 1.A.3(b) in \citeasnoun{shaked2007stochastic}, we have $S \geq_1 \underline{S}$ and the statement indeed follows.   Since $\underline{S} \sim$Binom$(N,\min_{i \in [N]}q_i)$, by Theorem 1 of   \citeasnoun{greenberg2014tight}  $\Pr\{\underline{S} > N \min_{i \in [N]}q_i\} > \frac{1}{4}$ if $\min_{i \in [N]}q_i > \frac{1}{N}$ and the result follows.
\end{proof} 

\begin{cor}\label{cor:1/4}
	Suppose $\{z_i\}_{i=1}^N$ are independently distributed, with $z_i \sim$Bernoulli$(q_i)$, $0 < q_i < 1$ for every $i \in [N]$.  Let $S := \sum_{i=1}^N z_i$, then $\Pr\{S < N \max_{i \in [N]}q_i\} > \frac{1}{4}$ if $ \max_{i \in [N]}q_i < 1- \frac{1}{N}$.  
\end{cor}

\begin{proof}
	Apply Lemma \ref{lem:1/4} to $\{\tilde{z}_i\}_{i=1}^N = \{1 - z_i\}_{i=1}^N$.
\end{proof}

We now extend \eqref{eq:relative1} and \eqref{eq:relative2} to accommodate non-identically (but independently) distributed random variables (i.ni.d), induced by conditioning. We first introduce additional notation as we now consider  ${\bf x} = (x_1,...,x_N)$ and ${\bf y} = (y_1,...,y_N)$ defined on the sample space $S^N$, independent of each other and both distributed according to  $\prod_{i=1}^N f(.|g_i)$ conditional on $g(i) = g_i, i \in [N]$. We use $\Pr\{\cdot|{\bf g}\}$ to denote the probability of event $\cdot$ conditional on $(g(1),...,g(N)) = \bf g $.
Define  
\begin{align}
	\underline{P}(A) &: = \int \underline{f}(x) dx, \\
	\overline{P}(A)& : = \int_A \overline{f}(x) dx. 
\end{align}
where $ \underline{f}$ and $\overline{f}$ are as defined in Section \ref{sec:assn}.
Define $\mathcal{B} = \{\mathcal{B}(x,\tau) \, | \, x\in \R^d, \, \tau >0\}$ and
\begin{align*}
	\underline{Q}& := \left \{(x_1,...,x_N) \in S^{N}: \exists A \in \mathcal C \text{ such that } \frac{\underline{P}(A) - P_{{\bf x} N}(A)}{\sqrt{\underline{P}(A)}} > \eta \right    \}, \\
	\overline{Q}& := \left \{(x_1,...,x_N) \in S^{N}: \exists A \in \mathcal C \text{ such that } \frac{ P_{{\bf x} N}(A) - \overline{P}(A)}{\sqrt{ P_{{\bf x} N}(A)}}> \eta \right    \}, \\
	R &:= \left \{(x_1,...,x_N,y_1,...,y_N) \in S^{2N}: \exists A \in \mathcal C \text{ such that } {P_{{\bf x} N}(A) - P_{{\bf y} N}(A)} > \eta{\sqrt{P_{{\bf xy} N}(A)}} \right    \}.  
\end{align*}

Lemma \ref{lem:extAST} and Lemma \ref{lem:extrelative2} extend \eqref{eq:relative1} and \eqref{eq:relative2} to i.ni.d data.

\begin{lem}\label{lem:extAST} For any collection $\mathcal{C}$ of Borel sets,
	 $$\Pr\{{\bf x} \in \underline{Q} |{\bf g}\} \leq 4 m_{\mathcal{C}}(2N)  \exp \left(-\eta^2 N / 4 \right).$$
\end{lem}

\begin{proof}
	The proof relies on two claims. 

\textbf{First claim:} $\Pr\{{\bf x} \in \underline{Q} |{\bf g}\} \leq 4 \Pr\{({\bf x} ,{\bf y}) \in R|{\bf g}\}$ if $N > 2/\eta^2$.  

Proof of the first claim:   We follow the proof of Theorem 2.1 of  \citeasnoun{anthony1993result}, which deals with IID sequences, while accomodating heterogeneity induced by conditioning on $\bf g$.    First note that for each  ${\bf x} \in \underline{Q}$ there exists a set $A_{\bf x} \in \mathcal{C}$, indexed by $\bf x$,  such that $\underline{P}(A_{\bf x}) - P_{{\bf x} N}(A_{\bf x}) > \eta \sqrt{\underline{P}(A_{\bf x})}.$    It then follows that $\inf_{{\bf x} \in \underline{Q}} \underline{P}(A_{\bf x}) \geq  \eta^2$ for ${\bf x} \in \underline{Q}$.  Define 
$$
F_{\bf xy}(A_{\bf x}) :=  \frac{P_{{\bf y} N}(A_{\bf x}) - P_{{\bf x} N}(A_{\bf x})}{\sqrt{P_{{\bf xy} N}(A_{\bf x})}}.  
$$
Note that for every ${\bf x} \in \underline{Q}$ 
\begin{eqnarray*}
	\min_{g \in [G]}\int_{A_{\bf x}} f(x|g)dx &\geq& \inf_{{\bf \xi} \in \underline{Q}} \underline{P}(A_{\bf \xi})
	\\
	&\geq& \eta^2
	\\
	&> & \frac{2}{N}.
\end{eqnarray*}
Therefore by Lemma \ref{lem:1/4}   we have $\Pr\{P_{{\bf y} N}(A_{\bf x})  >   \underline{P}(A_{\bf x}) |{\bf g}  \} > \frac{1}{4}$ for every $x \in \underline{Q}$.       
Then noting that $\frac{d}{dy} \left( (y-x) /   \sqrt{\frac{x+y}{2}}   \right)$ is nonnegative, 
\begin{eqnarray*}
	\Pr\{({\bf x},{\bf y}) \in R |{\bf g}\}   &\geq&    \Pr\{  F_{\bf xy}(A_{\bf x}) > \eta\ |  P_{{\bf y} N}(A_{\bf x})  >   \underline{P}(A_{\bf x}), {\bf x} \in \underline{Q}, {\bf g}  \} \Pr\{P_{{\bf y} N}(A_{\bf x})  >   \underline{P}(A_{\bf x}), {\bf x} \in \underline{Q} |{\bf g} \}
	\\
	&\geq& \Pr\left\{   \frac{\underline{P}(A_{\bf x}) - P_{{\bf x} N}(A_{\bf x})}{(\sqrt{(P_{{\bf x} N}(A_{\bf x})  +  \underline{P}(A_{\bf x}) )/2}} > \eta   |{\bf x} \in \underline{Q}, {\bf g}      \right \}   \left(   \inf_{{\bf x} \in \underline{Q}}     \Pr\{P_{{\bf y} N}(A_{\bf x})  >   \underline{P}(A_{\bf x}) |{\bf g} \}   \right)   \Pr\{  {\bf x} \in \underline{Q} |{\bf g} \}
	\\
	&\geq& \frac{1}{4}\Pr\left\{   \frac{ \eta \sqrt{\underline{P}(A_{\bf x})}    }{\sqrt{(P_{{\bf x} N}(A_{\bf x})  +  \underline{P}(A_{\bf x}) )/2}} > \eta | {\bf x} \in \underline{Q}, {\bf g}   \right \}         \Pr\{  {\bf x} \in \underline{Q} | {\bf g} \}
	\\
	&=&  \frac{1}{4}\ \Pr\{  {\bf x} \in \underline{Q} | {\bf g} \}.
\end{eqnarray*}

\textbf{Second claim:} if $N > 2/\eta^2$,  
$$ \Pr\{({\bf x},{\bf y}) \in R |{\bf g}\}    \leq m_{\mathcal{C}}(2N) \exp \left(-\eta^2 N / 4 \right).$$

Proof of the second claim:   Define $\Lambda$ as in \citeasnoun{anthony1993result}, i.e., the group generated by all transpositions of the form $(i , N+i)$ for $1\leq i \leq N$. Consider $\tau \in \Lambda$ and define for $z \in S^{2N}$, $\tau {\bf z }:= ({\bf z }_{\tau(1)},...,{\bf z }_{\tau(2N)})$. If $ {\bf x}$ and  ${\bf y}$ are two samples independent of each other and both distributed according to  $\prod_{i=1}^N f(.|g_i)$, then ${\bf xy}$ and $\tau {\bf xy}$ have the same distribution. Thus
\begin{eqnarray*}
	\Pr (R| {\bf g}) &=& \Pr ( \exists A_{\tau {\bf xy}} \in \mathcal C \text{ such that }  F_{\tau {\bf xy}}( A_{\tau {\bf xy}} ) > \eta \ | {\bf g})\\
	&=& \frac{1}{|\Lambda|}\ \E\left(\sum_{\tau \in \Lambda} \bold{1}\left[ \exists A_{\tau {\bf xy}} \in \mathcal C \text{ such that }  F_{\tau {\bf xy}}( A_{\tau {\bf xy}} ) > \eta  \right]  \, | {\bf g}\right)
\end{eqnarray*}
We use the notation $(A_{{\bf xy}}^{1},...,A_{ {\bf xy}}^{\Delta_{\mathcal{C}}( {\bf xy})})$ for a CSDR of $ {\bf xy}$. Note that any CSDR of $\tau {\bf xy}$ is a CSDR of ${\bf xy}$ and vice versa.
Then for all $A \in \mathcal{C}$, there exists $1 \leq t \leq \Delta_{\mathcal{C}}({\bf xy})$ such that $F_{\tau {\bf xy}}( A ) = F_{\tau {\bf xy}}( A_{{\bf xy}}^t )$. Thus 
\begin{eqnarray*}
	\Pr\{({\bf x},{\bf y}) \in R |{\bf g}\}  	 &\leq& \frac{1}{|\Lambda|}\ \E\left(\sum_{\tau \in \Lambda} \ \sum_{ t = 1 }^{\Delta_{\mathcal{C}}({\bf xy})} \bold{1}\left[ F_{\tau {\bf xy}}( A_{{\bf xy}}^t ) > \eta  \right]  \, | {\bf g}\right)
\end{eqnarray*}
Define $\Theta^t( {\bf xy})$ as in \citeasnoun{anthony1993result}, that is,  the number of permutations $\tau \in \Lambda$ such that $ F_{\tau {\bf xy}}( A_{{\bf xy}}^t ) > \eta $. The inequality above can be rewritten
\begin{eqnarray*}
	\Pr\{({\bf x},{\bf y}) \in R |{\bf g}\}    &\leq& \frac{1}{|\Lambda|}\ \E\left( \sum_{ t = 1 }^{\Delta_{\mathcal{C}}({\bf xy})} \sum_{\tau \in \Lambda} \bold{1}\left[ F_{\tau {\bf xy}}( A_{{\bf xy}}^t ) > \eta  \right]  \, | {\bf g}\right) = \frac{1}{|\Lambda|}  \ \E\left( \sum_{ t = 1 }^{\Delta_{\mathcal{C}}({\bf xy})}  \Theta^t( {\bf xy})   \, | {\bf g}\right).
\end{eqnarray*}
As in \citeasnoun{anthony1993result}, 
\begin{eqnarray*}
	\frac{\Theta^t( {\bf xy})}{|\Lambda|}  &\leq& \exp(-\eta^2 N /4),
\end{eqnarray*}
thus 
\begin{eqnarray*}
	\Pr\{({\bf x},{\bf y}) \in R |{\bf g}\}    &\leq& \Delta_{\mathcal{C}}({\bf xy}) \exp(-\eta^2 N /4) \leq m_{\mathcal{C}}(2N)  \exp \left(-\eta^2 N / 4 \right).
\end{eqnarray*}
By the first and second claim,  if $N > 2/\eta^2$,  $\Pr\{{\bf x} \in \underline{Q} |{\bf g}\} \leq 4 m_{\mathcal{C}}(2N)  \exp \left(-\eta^2 N / 4 \right)$. Note that if  $N \leq 2/\eta^2$, $4 \exp(-\eta^2 N /4) \geq 2$ thus the previous inequality also holds.
\end{proof}

The next lemma adapts Lemma \ref{lem:extrelative2} to i.ni.d data.

\begin{lem}\label{lem:extrelative2} For any collection $\mathcal{C}$ of Borel sets,
	$$\Pr\{{\bf x} \in \overline{Q} |{\bf g}\} \leq 4 m_{\mathcal{C}}(2N)  \exp \left(-\eta^2 N / 4 \right).$$
\end{lem}

\begin{proof} We look first at the case  $N > 2/\eta$ and prove that $\Pr\{{\bf x} \in \overline{Q} |{\bf g}\} \leq 4 \Pr\{({\bf x} ,{\bf y}) \in R|{\bf g}\}$. The rest of the proof follows by the second claim in the proof of Lemma \ref{lem:extAST}.
	
	First note that for each  ${\bf x} \in \overline{Q}$ there exists a set $A_{\bf x} \in \mathcal{C}$, indexed by $\bf x$,  such that $P_{{\bf x} N}(A_{\bf x}) - \overline{P}(A_{\bf x}) > \eta \sqrt{P_{{\bf x} N}(A_{\bf x})}$. As in the proof of Lemma \ref{lem:extrelative2}, this implies that $P_{{\bf x} N}(A_{\bf x}) > \overline{P}(A_{\bf x}) $ and $\overline{P}(A_{\bf x}) \leq 1 - \eta$.  Define 
	$$
	F_{\bf xy}(A_{\bf x}) :=  \frac{P_{{\bf x} N}(A_{\bf x}) - P_{{\bf y} N}(A_{\bf x})}{\sqrt{P_{{\bf xy} N}(A_{\bf x})}}.  
	$$
	If $ P_{{\bf y} N}(A_{\bf x})  <  \overline{P}(A_{\bf x})$, 
	\begin{align*}
		F_{\bf xy}(A_{\bf x}) &> \frac{ \overline{P}(A_{\bf x}) + \eta \sqrt{P_{{\bf x} N}(A_{\bf x})} - P_{{\bf y} N}(A_{\bf x})}{\sqrt{[P_{{\bf x} N}(A_{\bf x})) + P_{{\bf y} N}(A_{\bf x})]/2}} \\
		&> \frac{\eta \sqrt{P_{{\bf x} N}(A)} }{\sqrt{[P_{{\bf x} N}(A_{\bf x}) +  \overline{P}(A_{\bf x})]/2}} > \eta.
	\end{align*}	
	Note that for every ${\bf x} \in \overline{Q}$,
	$\overline{P}(A_{\bf x}) >\max_{g \in [G]} P(A_{\bf x}|g)$.
	Thus $ \Pr\{P_{{\bf y} N}(A_{\bf x})  <   \overline{P}(A_{\bf x}) |{\bf g} \}   \geq \Pr\{P_{{\bf y} N}(A_{\bf x})  <  \max_{g \in [G]} P(A_{\bf x}|g) |{\bf g} \} > 1/4 $ by Corollary \ref{cor:1/4} as long as $ \max_{g \in [G]} P(A_{\bf x}|g) \leq 1 - 1/N$. This holds by	 
	\begin{eqnarray*}
		\max_{g \in [G]} P(A_{\bf x}|g) &\leq& \overline{P}(A_{\bf x})
		\\
		&\leq& 1 - \eta
		\\
		&< & 1 - \frac{2}{N}.
	\end{eqnarray*}
	Therefore 
	\begin{eqnarray*}
		\Pr\{({\bf x},{\bf y}) \in R |{\bf g}\}   &\geq&    \Pr\{  F_{\bf xy}(A_{\bf x}) > \eta\ |  P_{{\bf y} N}(A_{\bf x})  <  \overline{P}(A_{\bf x}), {\bf x} \in \overline{Q}, {\bf g}  \} \Pr\{P_{{\bf y} N}(A_{\bf x}) <   \overline{P}(A_{\bf x}), {\bf x} \in \overline{Q} |{\bf g} \}
		\\
		&\geq&   \left(   \inf_{{\bf x} \in \overline{Q}}     \Pr\{P_{{\bf y} N}(A_{\bf x})  <   \overline{P}(A_{\bf x}) |{\bf g} \}   \right)   \Pr\{  {\bf x} \in \overline{Q} |{\bf g} \}
		\\
		&\geq&   \frac{1}{4}\ \Pr\{  {\bf x} \in \overline{Q} | {\bf g} \}.
	\end{eqnarray*}

The case  $N < 2/\eta$ is handled as in the proof of Lemma \ref{lem:extrelative2}.
\end{proof}
We now apply these results to derive bounds on radiuses. Lemma \ref{lem:extAST} implies the following version of Theorem 4 in \citeasnoun{portier2021nearest}.
\begin{cor}
	Let $(x(i))_{i \leq N}$ be a sequence of independent and nonidentically distributed random vectors valued in $\R^d$ and $\underline{P}$ and $\overline{P}$ defined as above. For any $\delta > 0$,
	
	With probability at least $ 1-\delta$ conditional on ${\bold g}$,
	\begin{equation}
		 \forall B \in \mathcal{B},	\frac{1}{N} \sum_{i=1}^N \bold{1} \left(x(i) \in B \right) \geq \underline{P}(B) \left( 1 - \sqrt{\frac{12 d \ln(12N/\delta)}{N \underline{P}(B)}}\right),  \label{eq:emp_lowerbd}
	\end{equation}
	
	With probability at least $ 1-\delta$ conditional on ${\bold g}$,
	\begin{equation}
		\forall B \in \mathcal{B},		\frac{1}{N} \sum_{i=1}^N \bold{1} \left(x(i) \in B \right) \leq \frac{12 d \ln(12N/\delta)}{N } + 4 \overline{P}(B).
		\label{eq:emp_upperbd}		
	\end{equation}
\end{cor}
\begin{proof}
	The proof of \eqref{eq:emp_lowerbd} applies Lemma  \ref{lem:extAST} to the collection $\mathcal{B}$ and follows the lines of the proof of Theorem 4 in \citeasnoun{portier2021nearest}.	
	
	To obtain \eqref{eq:emp_upperbd}, note that Lemma \ref{lem:extrelative2} applied to the collection $\mathcal{B}$ implies that  with probability at least $1-\delta$,
	$$\forall B \in \mathcal{B}, \frac{ P_{{\bf x} N}(B) - \overline{P}(B)}{\sqrt{ P_{{\bf x} N}(B)}} \leq \sqrt{\frac{4 [ \ln\left(4m_{\mathcal{C}}(2N)/\delta\right)]}{N } } \leq  \sqrt{\frac{12 d \ln(12N/\delta)}{N }} $$
	where the second inequality is obtained through the same arguments as in  the proof of Theorem 4 in \citeasnoun{portier2021nearest}. Define $\beta_n :=\sqrt{\frac{12 d \ln(12N/\delta)}{N }}$ and the function $g: x \mapsto x^2 - \beta_n x - \overline{P}(B)$.  The function $g$ has two roots, thus
	\begin{align*}
		&[P_{{\bf x} N}(B) - \overline{P}(B)]/\sqrt{ P_{{\bf x} N}(B)} \leq \beta_n \\
	&	\Rightarrow g(\sqrt{P_{{\bf x} N}(B)}) \leq 0 \\
	&	\Rightarrow  \sqrt{P_{{\bf x} N}(B)} \leq \frac{1}{2} \left(\beta_n + \sqrt{\beta_n^2 + 4\overline{P}(B)}\right) \leq \sqrt{\beta_n^2 + 4\overline{P}(B)}.
	\end{align*}
\end{proof}
We finally obtain the following upper and lower bounds on the radius.
\begin{lem}\label{lem:radius}
	Let Assumptions \ref{assn:boundsupport}, \ref{assn:boundinff} and \ref{assn:boundsupf} hold.  Then,
	\begin{enumerate}
		\item for all $N$, $\delta \in (0,1)$ and $1 \leq k \leq N$ such that $24 d \ln(12N/\delta) \leq k \leq T^d N b_X c V_d/2$,
		\begin{equation}\label{eq:rad_lowerbound}
			\Pr \left( \sup_{x \in S} r_k(x) \leq R_k | \bold g  \right) \geq 1-\delta,
		\end{equation}
		\item for all $N$, $\delta \in (0,1)$ and $1 \leq k \leq N$ such that $k \geq 12 d \ln(12N/\delta)$,
		\begin{equation}\label{eq:rad_upperbound}
			\Pr \left( \inf_{x \in S} r_k(x) \geq \underline{R}_k | \bold g  \right) \geq 1-\delta.
		\end{equation}
	\end{enumerate}
\end{lem}
\begin{proof}
	\eqref{eq:rad_lowerbound} is derived replacing $P$ with $\underline{P}$ in the proof of Lemma 4 in  \citeasnoun{portier2021nearest} and using \eqref{eq:emp_lowerbd}.
	
	To prove \eqref{eq:rad_upperbound}, note that 
	\begin{align*}
		&\overline{P}(\mathcal{B}(x,\underline{R}_k)) = \int_{\mathcal{B}(x,\underline{R}_k) \cap S} \overline{f}(x) dx \\
		&\leq \overline{U}_X \lambda(\mathcal{B}(x,\underline{R}_k) \cap S)\\
		& \leq \frac{ k - 12 d \ln(12N/\delta) }{4 N }
	\end{align*}
	 where $\lambda$ is the Lebesgue measure. By \eqref{eq:emp_upperbd},  with probability at least  $1-\delta$,
	 $$\forall x \in S,		\frac{1}{N} \sum_{i=1}^N \bold{1} \left(x(i) \in \mathcal{B}(x,\underline{R}_k) \right) \leq k/N$$
	which implies  \eqref{eq:rad_upperbound}.
\end{proof}

\subsection{Results on Laplacian}

\

\begin{lem}\label{lem:normL} 
	$\| \widetilde{L^\eta_\tau}  \|  \leq 1.$
\end{lem}

\begin{proof}
	By Theorem 7.3.3 of \citeasnoun{horn2012matrix}, the eigenvalues $\widetilde{L^\eta_\tau}$ are $\sigma_1(L^\eta_\tau) \geq ... \geq \sigma_G(L^\eta_\tau) \geq 0 \geq -\sigma_G(L^\eta_\tau) \geq ... \geq -\sigma_1(L^\eta_\tau)$. Thus the claim holds if $\sigma_1(L^\eta_\tau) \leq 1$.
	We equivalently show that $I - \widetilde{L^\eta_\tau}$ is a symmetric  positive semidefinite matrix.
	Note that
	$$I - \widetilde{L^\eta_\tau} = 
	\begin{pmatrix}
		I & -L^\eta_\tau\\
		-\left(L^\eta_\tau\right)^\top & I
	\end{pmatrix}
	.$$
	Take $c \in \R^{2k}$, write $c = 
	\begin{pmatrix}
		y\\
		z
	\end{pmatrix}
	$ where $y,z \in \R^{k}$.
	Then
	\begin{align*}
		c^{\top} \left[I - \widetilde{L^\eta_\tau}\right] c& = y^{\top} y + z^{\top}z - 2 y^{\top}L^\eta_\tau z \\
		&= \sum_{i \in \eta_N(x)} y_i^2 +  \sum_{j \in \eta_N(x')} z_j^2 - 2 \sum_{\substack{i \in \eta_N(x) \\ j \in \eta_N(x')}} \frac{y_i z_j A_{ij}}{\left[\left(O_\tau^\eta\right)_{ii} \left(Q_\tau^\eta\right)_{jj}\right]^{1/2}}\\
		& \geq\sum_{\substack{i \in \eta_N(x) \\ j \in \eta_N(x')}} \left( \frac{y_i A_{ij}}{\left(O_\tau^\eta\right)_{ii}^{1/2}} -  \frac{z_j A_{ij}}{\left(Q_\tau^\eta\right)_{jj}^{1/2}}  \right)^2 \geq 0,
	\end{align*}
	where the first inequality comes from $\tau \geq 0$ thus guaranteeing $I - \widetilde{L^\eta_\tau}$ is positive semidefinite.
\end{proof}

\begin{lem}\label{lem:normmathcalL} 
	$\|  \widetilde{{\mathcal L}_\tau^\eta}({\bf x},{\bf g})   \| \leq 1.$
\end{lem}	

\begin{proof}  
	We proceed as for $L_\tau^\eta$: we can equivalently show that $I - \widetilde{{\mathcal L}_\tau^\eta}({\bf x},{\bf g})  $  is a symmetric  positive semidefinite matrix.
	Take $c \in \R^{2k}$, write $c = 
	\begin{pmatrix}
		y\\
		z
	\end{pmatrix}
	$ where $y,z \in \R^{k}$.
	Then note that		
	\begin{eqnarray*}
		&& \sum_{\substack{i \in \eta_N(x) \\ j \in \eta_N(x')}} 
		\left( \frac{y_i \sqrt{P^\eta_{ij}({\bf x},{\bf g}) } }{\left[ \left(\mathcal{O}_\tau^\eta({\bf x},{\bf g})\right)_{ii}  \right]^{1/2}}
		- 
		\frac{z_j \sqrt{P^\eta_{ij}({\bf x},{\bf g}) } }{\left[ \left(\mathcal{Q}_\tau^\eta({\bf x},{\bf g})\right)_{jj}  \right]^{1/2}}
		\right)^2 
		\\
		& =&  \sum_{i \in \eta_N(x)}  \sum_{j \in \eta_N(x')}   \frac{y_i^2 P^\eta_{ij}({\bf x},{\bf g}) }{ \left(\mathcal{O}_\tau^\eta({\bf x},{\bf g})\right)_{ii}} 
		+ 	 \sum_{j \in \eta_N(x')}  \sum_{i \in \eta_N(x)}   \frac{z_j^2 P^\eta_{ij}({\bf x},{\bf g}) }{ \left(\mathcal{Q}_\tau^\eta({\bf x},{\bf g})\right)_{jj}   } 
		- 2  \sum_{\substack{i \in \eta_N(x) \\ j \in \eta_N(x')}}	 \frac{y_i z_j P^\eta_{ij}({\bf x},{\bf g})  }{\left[ \left(\mathcal{O}_\tau^\eta({\bf x},{\bf g})\right)_{ii}   \left(\mathcal{Q}_\tau^\eta({\bf x},{\bf g})\right)_{jj}  \right]^{1/2} }
		\\
		&\leq& 
		\sum_{i \in \eta_N(x)} y_i^2 +  \sum_{j \in \eta_N(x')} z_j^2 - 2 \sum_{\substack{i \in \eta_N(x) \\ j \in \eta_N(x')}} \frac{y_i z_j  P^\eta_{ij}({\bf x},{\bf g})  }{\left[ \left(\mathcal{O}_\tau^\eta({\bf x},{\bf g})\right)_{ii}   \left(\mathcal{Q}_\tau^\eta({\bf x},{\bf g})\right)_{jj}  \right]^{1/2}},
		\end{eqnarray*}
	by $\left(\mathcal{O}_\tau^\eta({\bf x},{\bf g})\right)_{ii} = \sum_{j \in \eta_N(x')} P^\eta_{ij}({\bf x},{\bf g}) + \tau$,  $\left(\mathcal{Q}_\tau^\eta({\bf x},{\bf g})\right)_{jj} = \sum_{i \in \eta_N(x)} P^\eta_{ij}({\bf x},{\bf g}) + \tau$ and $\tau >0$. Thus
	we obtain
	\begin{align*}
			c^{\top} \left[I - \widetilde{{\mathcal L}_\tau^\eta}({\bf x},{\bf g})  \right] c& = y^{\top} y + z^{\top}z - 2 y^{\top}  {\mathcal L}_\tau^\eta({\bf x},{\bf g})   z \\
			& = \sum_{i \in \eta_N(x)} y_i^2 +  \sum_{j \in \eta_N(x')} z_j^2 - 2 \sum_{\substack{i \in \eta_N(x) \\ j \in \eta_N(x')}} \frac{y_i z_j  P^\eta_{ij}({\bf x},{\bf g})  }{\left[ \left(\mathcal{O}_\tau^\eta({\bf x},{\bf g})\right)_{ii}   \left(\mathcal{Q}_\tau^\eta({\bf x},{\bf g})\right)_{jj}  \right]^{1/2}}
			\\
			& \geq \sum_{\substack{i \in \eta_N(x) \\ j \in \eta_N(x')}} 
			\left( \frac{y_i \sqrt{P^\eta_{ij}({\bf x},{\bf g}) } }{\left[ \left(\mathcal{O}_\tau^\eta({\bf x},{\bf g})\right)_{ii}  \right]^{1/2}}
			- 
			\frac{z_j \sqrt{P^\eta_{ij}({\bf x},{\bf g}) } }{\left[ \left(\mathcal{Q}_\tau^\eta({\bf x},{\bf g})\right)_{jj}  \right]^{1/2}}
			\right)^2  	\geq 0.
		\end{align*}
\end{proof}

\clearpage
\bibliographystyle{econometrica}
\bibliography{mybib}
\end{document}